\documentclass[american,french,english]{paper}
\usepackage[T1]{fontenc}
\usepackage[latin9]{inputenc}
\usepackage{float}
\usepackage{textcomp}
\usepackage{amsthm}
\usepackage{multirow}
\usepackage{amsmath}
\usepackage{amssymb}
\usepackage{fixltx2e}
\usepackage{graphicx}
\usepackage{esint}
\usepackage{setspace}

\makeatletter

\providecommand{\tabularnewline}{\\}

 \newcommand\thmsname{\protect\theoremname}
 \newcommand\nm@thmtype{theorem}
 \theoremstyle{plain}
 
 \newenvironment{namedthm}[1][Undefined Theorem Name]{
   \ifx{#1}{Undefined Theorem Name}\renewcommand\nm@thmtype{theorem*}
   \else\renewcommand\thmsname{#1}\renewcommand\nm@thmtype{namedtheorem}
   \fi
   \begin{\nm@thmtype}}
   {\end{\nm@thmtype}}
\theoremstyle{plain}
\newtheorem{thm}{\protect\theoremname}[section]
  \theoremstyle{remark}
  \newtheorem{rem}[thm]{\protect\remarkname}
  \theoremstyle{plain}
  \newtheorem{prop}[thm]{\protect\propositionname}
  \newtheorem{lem}[thm]{\protect\lemmaname}

\date{October, 21rst 2014}
\makeatother

\usepackage{babel}
  \providecommand{\lemmaname}{Lemma}
\makeatletter
\addto\extrasfrench{%
   \providecommand{\fg}{\ifdim\lastskip>\z@\unskip\fi~\frqq}%
}

\makeatother
  \addto\captionsamerican{\renewcommand{\propositionname}{Proposition}}
  \addto\captionsamerican{\renewcommand{\remarkname}{Remark}}
  \addto\captionsamerican{\renewcommand{\theoremname}{Theorem}}
  \addto\captionsenglish{\renewcommand{\propositionname}{Proposition}}
  \addto\captionsenglish{\renewcommand{\remarkname}{Remark}}
  \addto\captionsenglish{\renewcommand{\theoremname}{Theorem}}
  \addto\captionsfrench{\renewcommand{\propositionname}{Proposition}}
  \addto\captionsfrench{\renewcommand{\remarkname}{Remarque}}
  \addto\captionsfrench{\renewcommand{\theoremname}{Th�or�me}}
  \providecommand{\propositionname}{Proposition}
  \providecommand{\remarkname}{Remark}
  \providecommand{\theoremname}{Theorem}
\providecommand{\theoremname}{Theorem}

\begin{document}

\title{State and Parameter Estimation of Partially Observed Linear Ordinary Differential Equations with Deterministic Optimal Control}

 \author{{\large  Quentin Clairon,  Nicolas J-B. Brunel}}
\institution{ENSIIE \& Laboratoire de Math\'ematiques et Mod\'elisation d'Evry,\\  UMR CNRS 8071, Universit\'e
d'Evry, France}

\maketitle
\begin{abstract}
Ordinary Differential Equations are a simple but powerful framework
for modeling complex systems. Parameter estimation from times series
can be done by Nonlinear Least Squares (or other classical approaches),
but this can give unsatisfactory results because the inverse problem can be ill-posed, even when the differential equation is linear.  

Following recent approaches that use approximate solutions of the
ODE model, we propose a new method that converts parameter estimation
into an optimal control problem: our objective is to determine a control
and a parameter that are as close as possible to the data. We derive
then a criterion that makes a balance between discrepancy with data
and with the model, and we minimize it by using optimization in functions
spaces: our approach is related to the so-called Deterministic Kalman Filtering,
but different from the usual statistical Kalman filtering. 

We show the root-$n$ consistency and asymptotic normality of the
estimators for the parameter and for the states. Experiments in a
toy model and in a real case shows that our approach is generally
more accurate and more reliable than Nonlinear Least Squares and Generalized
Smoothing, even in misspecified cases. 

\end{abstract}

\section{Introduction}

Ordinary Differential Equations (ODE) are a widely used class of mathematical
models in biology, physics, engineering, \dots Indeed, it is a relatively
simple but powerful framework for expressing the main mechanisms and
interactions of potentially complex systems. It is often a reference
framework in population dynamics and epidemiology \cite{Ellner2006},
virology \cite{Nowak2000}, or in genetics for describing gene regulation
networks \cite{Marbach2012,Wu2014}. The model takes the form $\dot{x}=f(t,x,\theta)$,
where $f$ is a vector field, $x$ is the state, and $\theta$ is
a parameter that can be partly known. The parameter $\theta$ is often
of high interest, as it represents rates of changes, phenomenological
constants needed for interpretability and analysis of the system.
Typically, $\theta$ can be related to the sensitivity of a variable
with respect to other variables. 

Hence, the parameter estimation of ODEs from experimental data is
a long-standing statistical subject that have been adressed with many
different tools. Estimation can be done with classical estimators
such as Nonlinear Least Squares (NLS) and Maximum Likelihood Estimator
(MLE) \cite{LiOsborne2005,walter1997identification,Pronzato2008}
or Bayesian approaches \cite{Huang2006,gelman1996,ghasemi2011,calderhead2009}.
Nevertheless, the statistical estimation of an ODE model by NLS leads
to a difficult nonlinear estimation problem. Some difficulties were
pointed out by Ramsay et al. \cite{Ramsay2007} such as computational
complexity, due to ODE integration and nonlinear optimization. These
difficulties are in fact reminiscent of intrinsic difficulties in
the parameter estimation problem, that makes it an ill-posed inverse
problem, that needs some regularization \cite{Engl2009,Stuart2010}
. 

Alternative statistical estimators have been developped to deal with
this particular framework, such as Generalized Smoothing \cite{Ramsay2007,QiZhao2010,Campbell2011,Campbell2013}
or Two-Step estimators \cite{Varah1982,Brunel2008,LiangWujasa2008,GugushviliKlaassen2010,Brunel2014}.
Two-step estimators use a nonparametric estimator $\hat{X}$ and aim
at minimizing quantities characterizing the differential models, such
as the weighted $L^{2}$ distance $\int_{0}^{T}\left\Vert \dot{\widehat{X}}(t)-f(t,\widehat{X}(t),\theta)\right\Vert ^{2}w(t)dt$.
These estimators have a good computational efficiency as they avoid
repeated ODE integration. In practice, the used criteria are also
smoother and easier to optimize than the NLS criterion. Two-step estimators
are consistent in general, but there is a trade-off with the statistical
precision, and some care in the use of nonparametric estimate $\dot{\widehat{X}}$
has to be taken in order to keep a parametric rate \cite{Brunel2008,GugushviliKlaassen2010}.

In the case of Generalized Smoothing \cite{Ramsay2007}, the solution
$X^{*}$ is approximated by a basis expansion that solves approximately
the ODE model; hence, the parameter inference is performed by dealing
with an imperfect model. Based on
the Generalized Profiling approach, Hooker proposed a criteria that
estimates the lack-of-fit through the estimation of a ``forcing function''
$t\mapsto u(t)$ in the ODE $\dot{x}-f(t,x,\tilde{\theta})=u(t)$,
where $\tilde{\theta}$ is a previous estimate obtained by Generalized
Profiling. 

In \cite{BrunelClairon_Riccati2014}, the authors have proposed a
two-step estimator for linear models, that avoids the use of $\dot{\widehat{X}}$
and introduces a forcing function without the finite basis decomposition
by using control theory. The principle is to transform the estimation problem
into a control problem: we have to find the best (or smallest) control
$u$ such that the ODE is close to the data. The limitations of the
results provided in \cite{BrunelClairon_Riccati2014} were the restriction
to fully observed system with known initial condition. The objective
of this paper is to provide a similar two-step estimate that permits
the estimation of $\theta$ without knowing $x_{0}$, that deals with
the partially observed case and provides state estimates. 

One interest of the approach used is to deal directly with the optimization
in a function space without using of series expansion for function
estimation. Moreover, infinite dimensional optimization tools give
a powerful characterization of the solutions, useful in practice.
This work can be seen as an extension of the previous one \cite{BrunelClairon_Riccati2014},
aiming to use control theory result for parameter inference. We deal
now with the partially observed case with unknown initial condition, 
that gives rise to a methodology close to the so-called ``Deterministic Kalman Filter''.
Indeed, in that paper, we assume that the system is linear, with a
linear observation function. 

Our method provides a consistent parametric estimator when the model
is correct. We show that it is root-n consistent and asymptotically
normal. At the same time, we get a discrepancy measure between the
model and the data under the form of an optimal control $u$ analogous
to the forcing function in \cite{Hooker2009}, and we show that we
can estimate the final and initial conditions and hence all the states
if needed, in particular the hidden ones. 

In the next section, we introduce the notations and we motivate our
approach by discussing the Generalized Smoothing approach, and the
link with Optimal Control Theory. In section \ref{sec:Estimator_properties},
we investigate the existence and regularity of our new criterion;
in particular, we derive necessary and sufficient conditions for defining
our approach in partially observed case. We show that the estimator is consistent under some regularity assumption
about the model. Then in section \ref{sec:Asymptotic_Normality},
we show that we reach the root$-n$ rate using regression splines
for $\widehat{Y}$ the nonparametric estimator of the observed signal.
We derive then the consistency of the state estimator derived.  Finally,
we show the interest of our method on a toy model and in a real model
used in chemical engineering, by a comparison with Nonlinear Least
Squares and Generalized Smoothing.

\section{Model and methodology\label{sec:Model-and-methodology}}

We introduce first the statistical ODE model of interest, and the basic notations for defining our estimator. We  relate this work to the Generalized Smoothing  estimator and the Tracking estimator. 

\subsection{Model and Notations}

We partially observe a ``true'' trajectory $X^{*}$ at random
times $0=t_{1}<t_{2}\dots<t_{n}=T$, such that we have $n$ observations
$(Y_{1},\dots,Y_{n})$ defined as \[
Y_{i}=CX^{*}(t_{i})+\epsilon_{i}
\]
where $\epsilon_{i}$ is a random noise and $C$ is the observation matrix of size $d'\times d$.  

We assume that there is a true parameter $\theta^{*}$ belonging to a subset $\Theta$
of $\mathbb{R}^{p}$, such that $X^{*}$ is the unique solution of
the linear ODE 
\begin{equation}
\dot{x}(t)=A_{\theta}(t)x(t)+r_{\theta}(t)\label{eq:LinearODEmodel}
\end{equation}
with initial condition $X^{*}(0)=x_{0}^{*}$; where $t\mapsto A_{\theta}(t)\in\mathbb{R}^{d\times d}$
and $t\mapsto r_{\theta}(t)\in\mathbb{R}^{d}$. More generally, we
denote $X_{\theta,x_{0}}$ the solution of (\ref{eq:LinearODEmodel})
for a given $\theta$, and initial condition $x_{0}$. We assume that
 $x_{0}^{*}$ and  $\theta^{*}$ are unknown, and that they must be estimated  from the data $(y_{1},\dots,y_{n})$. The parameter $\theta^{*}$ is the main parameter of interest, whereas the initial condition is considered as a nuisance parameter, needed essentially for the computation of candidate trajectories $X_{\theta, x_0}$. 

For linear equations, a central role is played by the solutions of the homogeneous ODE
\begin{equation}
\dot{x}(t)=A_{\theta}(t)x(t). \label{eq:ParametricHomogenousLinearODEmodel}
\end{equation}
Indeed, for each $s$ in $[0,T]$, we denote $t \mapsto \Phi_{\theta}\left(t,s\right)$ the solution to the matrix ODE (\ref{eq:ParametricHomogenousLinearODEmodel}),    
with initial condition $I_d$ at time $s$ (i.e  $\Phi_{\theta}\left(s,s\right)=I_d$). The function $(t,s) \mapsto \Phi_{\theta}\left(t,s\right)$ is a  $d\times d$ matrix valued function, called the resolvant of the ODE. It permits to give an explicit dependence of the solutions of (\ref{eq:LinearODEmodel}) in $r_\theta$ and the initial condition $x_0$, thanks to Duhamel's formula: 
\[
X_{\theta,x_{0}}(t)=\Phi_{\theta}(t,0)x_{0}+\int_{0}^{t}\Phi_{\theta}(t,s)r_{\theta}(s)ds.
\]
A consistent and classical method for the estimation of  $\theta^{*}$ is 
Nonlinear Least Squares (NLS),  that minimizes 
\[
\sum_{i=1}^{n}\left\Vert Y_{i}-CX_{\theta,x_{0}}(t_{i})\right\Vert _{2}^{2}. 
\]
A classical alternative is Generalized Smoothing (GS), that uses approximate solutions of the ODE (\ref{eq:LinearODEmodel}). GS replaces the solutions $X_{\theta,x_0}$ by splines that smooth data and solve approximately the ODE with a penalty based on the ODE model.  A basis expansion $\widehat{X}(t,\theta)=\widehat{\beta}(\theta)^{T}p(t)$
is computed for each $\theta$, where $\hat{\beta}(\theta)$ is obtained
by minimizing in $\beta$ the criterion 
\begin{equation}
J_{n}(\beta\vert\theta,\lambda)=\sum_{i=1}^{n}\left\Vert Y_{i}-C\beta^{T}p(t)\right\Vert _{2}^{2}+\lambda\int_{0}^{T}\left\Vert \beta^{T}\dot{p}(t)-\left(A_{\theta}(t)\beta^{T}p(t)+r_{\theta}(t)\right)\right\Vert _{2}^{2}dt\label{eq:Cost_GS}
\end{equation}
This first step is considered as profiling along the nuisance parameter
$\beta$, whereas the estimation of the parameter of interest is obtained
by minimizing the sum of squared errors of the proxy $\hat{X}(t,\theta)$:
\begin{equation}
\hat{\theta}^{GS}=\arg\min_{\theta}\sum_{i=1}^{n}\left\Vert Y_{i}-C\hat{X}(t_{i},\theta)\right\Vert ^{2}\label{eq:Estimator_GS}
\end{equation}
In practice, the hyperparameter $\lambda$ needs to be selected from the data with adaptive procedures, see \cite{Chernovena2014}. 

The essential difference with NLS is the replacement of the exact solution $X_{\theta,x_{0}}$
by the approximation $\hat{X}(\cdot,\theta)$ (that depends also on
the data). This change
induces a new source of error in the estimation of the true trajectory
$t\mapsto X^{*}(t)$ as the functions $\hat{X}(\cdot,\theta)$
are splines that do not solve exactly the ODE model (\ref{eq:LinearODEmodel}).
The ODE constraint is relaxed into an inequality constraint defined on the interval
$[0,T]$. 
The model constraint is never set to 0 because of the
trade-off with the data-fitting term $\sum_{i=1}^{n}\left\Vert Y_{i}-C\beta^{T}p(t)\right\Vert _{2}^{2}$.
For this reason, the ODE model (\ref{eq:LinearODEmodel}) is not solved
and it is useful to introduce the discrepancy term $\hat{u}_{\theta}(t)=\beta^{T}\dot{p}(t)-\left(A_{\theta}(t)\beta^{T}p(t)+r_{\theta}(t)\right)$
that corresponds to a model error. In fact, the proxy $\hat{X}(\cdot,\theta)$
satisfies the perturbed ODE $\dot{x}=A_{\theta}x+r_{\theta}+\hat{u}_{\theta}$.
This forcing function $\hat{u}_{\theta}$ is an outcome of the optimization
process and can be relatively hard to analyze or understand, but its analysis provides a good insight into the relevancy of the model \cite{Hooker2009,HookerEllner2013}. 

Based on these remarks, we introduce the perturbed linear ODE  \begin{equation}
\dot{x}(t)=A_{\theta}(t)x(t)+r_{\theta}(t)+u(t)\label{eq:ControlLinearODEmodel-1}
\end{equation} where the function $t\mapsto u(t)$ can be any function in $L^{2}$.
The solution of the corresponding Initial Value Problem 
\[
\begin{cases}
\dot{x}(t) & =A_{\theta}(t)x(t)+r_{\theta}(t)+u(t) \\
x(0) & =x_{0} 
\end{cases}
\]
is denoted $X_{\theta,x_{0},u}$. Instead of using the spline proxy
$\hat{X}(\cdot,\theta)$ for approximating $X^{*}$, we use the trajectories
$X_{\theta,x_{0},u}$ of the ODE (\ref{eq:ControlLinearODEmodel-1})
controlled by the additional functional parameter $u$. 

In \cite{BrunelClairon_Riccati2014},  the same perturbed model is introduced but the cost function is simpler as the observation matrix $C$ is the identity, and the initial condition is fixed. In that framework, an M-estimator for $\theta$ is proposed, based on the optimization of the criterion 
\begin{equation}
\tilde{S}(\widehat{Y} ; x_0, \theta, \lambda) = \inf_{u \in L^2} \{ \Vert \widehat{Y} - X_{\theta,x_{0},u}\Vert_{L^2}^2 + \Vert u \Vert_{L^2}^2 \}.  \label{eq:S_Tracking_criterion}
\end{equation}
The proper definition of $\tilde{S}$ and the derivation of its properties were obtained by using some classical results of Optimal Control Theory. Essentially, the computation of $\tilde{S}$ corresponds to the classical "tracking problem" that can be solved by the Linear-Quadratic theory (LQ theory). LQ theory solves the minimization problem in $L^2$ of the cost function  
\begin{equation}
C(u)  = \Vert  X_{\theta,x_{0},u}(t)\Vert_{L^2}^2 + \Vert u(t)\Vert_{L^2}^2  +  X_{\theta,x_{0},u}(T)^\top Q X_{\theta,x_{0},u}(T) \\
 \label{eq: QuadraticCost}
\end{equation}
The criteria $\tilde{S}$ used for parameter estimation is associated to the \emph{value function} defined in Optimal Control as $S(t,x)=\inf \{ C(u) \vert X_{\theta,x,u}(t)=x \} $. The value function plays a critical role in the analysis of optimal control problems, typically for the computation of an optimal policy. Under regularity assumptions, the value function $S$ is the solution of the Hamilton-Jacobi-Bellman Equation, which is a first order Partial Differential Equation \cite{Bardi2008}. Quite remarkably, for a linear ODE with a quadratic cost such as (\ref{eq: QuadraticCost}), the value function is a quadratic form in the state $x$, i.e $S(t,x)= - x^\top E(t) x$, where $E(t)$ is the solution of a matrix ODE (the Riccati equation), which makes its computation  very tractable in practice.\\

LQ theory can be adapted for tracking of an output signal $\widehat{Y}= C X^* + \epsilon$ with a perturbed linear ODE, see chapter 7 in \cite{Sontag1998}.  When we do not know the initial condition, some adaptations are required. Indeed, as the initial condition can have a strong influence on the optimal control and the optimal cost; it seems much harder to solve the control problem when the initial condition is not known: the current state $x(t)$ is unknown and all the admissible trajectories must be considered. Nevertheless, this problem is solved by the Deterministic Kalman Filter (DKF) by using the fact that the value function $S$ is a quadratic form on the state. \\
We show in the next section that the Deterministic Kalman Filtering (DKF) is well adapted for developing parameter estimation, as it enables to profile on $x_0$, considered as a nuisance parameter. In a two-step approach, it is critical as we need to control the influence of the nonparametric estimate of $\widehat{Y}$ on the convergence rate. As we use $\widehat{Y}(0)$ as a proxy for $Cx_0^*$, we need to show that the rate of the two-step estimator is not polluted by the use of nonparametric estimates of the boundary conditions, and that we keep a parametric rate for $\theta^*$ and $x_0^*$. This property was carefully checked in \cite{Brunel2008, Brunel2014,LiangWujasa2008}; in that paper, as we do not use implicitly or explicitly the derivative of the nonparametric estimate, the mechanics of the proof are different. 

In the next section, we give some details on LQ theory and on the criterion $S$. The classical costs in optimal control consist of an integral term plus a penalty term on the final state, such as $ X_{\theta,x_{0},u}(T)^\top Q X_{\theta,x_{0},u}(T)$. A preliminary time-reversing transformation is used for introducing properly the initial state in the cost $C$, rather than the final state. In a second step, we derive the criterion $S$, and we give a tractable expression for estimation.  Finally, we discuss the importance of identifiability and observability in the definition on our criterion. 

\subsection{The Deterministic Kalman Filter and the profiled cost}

Following the Tracking estimator, we look for a candidate
$X_{\theta,x_{0},u}$ that minimizes at the same time the discrepancy
with the data and the size of the perturbations $\left\Vert u\right\Vert _{L^{2}}$.  We consider nearly the same cost as in \cite{BrunelClairon_Riccati2014}
\begin{equation}
\tilde{C}\left(\hat{Y};x_{0},u,\theta,\lambda\right)=\int_{0}^{T}\left\Vert \hat{Y}(t)-CX_{\theta,x_{0},u}(t)\right\Vert _{2}^{2}dt+\lambda\int_{0}^{T}\left\Vert u(t)\right\Vert _{2}^{2}dt \label{eq:kalman_cost_function}
\end{equation}
for given $\lambda>0$. We can also add a positive quadratic form $x_{0}^{T}Qx_{0}$, where $Q$ is a positive symmetric matrix $Q$. This additional term permits to introduce easily some prior knowledge on $x_0$ such that we have a cost defined as 
\[
C\left(\hat{Y};u,x_{0},\theta,\lambda\right)  = x_{0}^{T}Qx_{0} +  \tilde{C}\left(\hat{Y};x_{0},u,\theta,\lambda\right). 
\]
Moreover, the matrix $Q$ avoids some technical problems in the definition of our criterion $S$. \\
For each $\theta$ in $\Theta$, we denote 
\begin{equation}
S\left(\hat{Y};\theta,\lambda\right)=\inf_{x_{0}} \left\{ \inf_{u\in L^{2}} C\left(\hat{Y};x_{0},u,\theta,\lambda\right) \right\} \label{eq:ProfiledCost}
\end{equation}
obtained by ``profiling'' on the function $u$ and then in the initial condition
$x_{0}$. The function $\tilde{S}\left(\hat{Y};x_0,\theta,\lambda \right) = \inf_{u\in L^{2}}C\left(\hat{Y};x_{0},u,\theta,\lambda\right) $ is the criterion used in the case of fixed and known initial conditions. Our approach is rather "natural" as we simply profile the regularized criterion  $x_{0}^{T}Qx_{0} + \tilde{S}\left(\hat{Y};x_0,\theta,\lambda\right)$. 

The definition of $S$ mimics the minimization of $J_{n}(\beta\vert\theta,\lambda)$ except that GS uses a discretized solution, defined on a B-splines basis. Nevertheless, our estimator possesses two other essential differences with Generalized Smoothing. As it was already mentioned in \cite{BrunelClairon_Riccati2014}, we define our estimator as the global minimum of  the profiled cost:
\begin{equation}
\widehat{\theta}^{K}=\arg\min_{\theta\in\Theta}S\left(\hat{Y};\theta,\lambda\right)\label{eq:OurEstimator}
\end{equation}
whereas the GS estimator minimizes a different criterion $\sum_{i=1}^{n}\left\Vert Y_{i}-C\hat{X}(t_{i},\theta)\right\Vert ^{2}$.
This means that in our methodology, we try to find a parameter $\theta$
that maintain a reasonable trade-off between model and data, whereas
the Generalized Smoothing Estimator $\hat{\theta}^{GS}$ is dedicated
to fit the data with the proxy $\hat{X}(\cdot,\theta)$, without considering
the size of model error represented by $\bar{u}_{\theta}$. 
Another important difference is in the way we deal with the unobserved part of the system. For simplicity, let us consider that we observe only the first $k<p$ components of $X$, such that the state vector can be written $X=\left( X^{obs}, X^{unobs} \right)$. For Generalized Smoothing, both functions $X^{obs}$ and $X^{unobs}$ are decomposed in a B-spline basis, and the corresponding coefficients $\beta^{obs}$ and $\beta^{unobs}$ are obtained by minimizing $J_n\left(\beta^{obs},\beta^{unobs} \vert \theta, \lambda \right)$. Because $\beta^{unobs} $ does not have to make a trade-off between the data and the ODE model, the estimated missing part $\hat{X}^{unobs}(\cdot,\theta)$ is the exact solution to the ODE (\ref{eq:ControlLinearODEmodel-1}). 
At the contrary, even in the case of partial observations, the perturbed solution $X_{\theta,x_0,u}$ is used for estimating the missing states and a perturbation exists for each component. Consequently, the estimated hidden states are not solution of the initial ODE. We think that this an advantage for state and parameter estimation with respect to Generalized Smoothing (and NLS) because it avoids to rely too strongly on a uncertain model during estimation. This uncertainty can be caused by errors in parameter estimation, or it can be due model misspecification, such as the presence of a forcing function $u^*$. In our experiments, we show that imposing model uncertainty for the unobserved variables is beneficial for error prediction. \\

Before going deeper into the interpretation and analysis of our estimator,
we need to show that the criterion $S\left(\hat{Y};\theta,\lambda\right)$
is properly defined and that we can obtain a tractable expression
for computations and for the theoretical analysis of (\ref{eq:OurEstimator}).
We use the Deterministic Kalman Filter (DKF) to obtain a closed-form expression for the minimal cost w.r.t the control
$u$ and $x_{0}$ (\ref{eq:ProfiledCost}). 

The initial aim of the DKF is to propose an
estimation of the final state $X^{*}(T)$ by making a balance between
the information brought by the noisy signal $\widehat{Y}$ and the
ODE model (see \cite{Sontag1998} for an introduction). We recall the two steps necessary for the filter construction, more
details are given in appendix:
\begin{enumerate}
\item For a given initial condition $x_{0}$, we find the minimum
cost thanks the fundamental theorem in LQ Theory (presented in \ref{thm:LQ_existence_unicity}), 
\item We minimize the quadratic form w.r.t the final
condition. 
\end{enumerate}
We give now the main theorem of that section about the existence of the criterion defined in equation (\ref{eq:ProfiledCost}). 

\begin{namedthm}[Theorem and Definition of $S\left(\zeta;\theta,\lambda\right)$]
\label{thm:kalman_existence_unicity} Let $t\mapsto\zeta(t)$ be a function belonging to $L^{\infty}(\left[0,\, T\right],\mathbb{R}^{d'})$
and $X_{\theta,x_{0},u}$ be the solution to the controlled ODE (\ref{eq:ControlLinearODEmodel-1}).
\\
For any $\theta$ in $\Theta$, $\lambda>0$, $Q > 0$, there exists a unique optimal control
$\bar{u}_{\theta,\lambda}$ and initial condition $\widehat{x_{0}}$
that minimizes the cost function 
\begin{equation}
C\left(\zeta;u,x_{0},\theta,\lambda\right)=x_{0}^{T}Qx_{0}+\int_{0}^{T}\left\{ \left\Vert \zeta(t)-CX_{\theta,x_{0},u}(t)\right\Vert _{2}^{2}+\lambda\left\Vert u(t)\right\Vert _{2}^{2}\right\} dt\label{eq:LQCost}
\end{equation}
The optimal control $\bar{u}_{\theta,\lambda}$ is
\begin{equation}
\bar{u}_{\theta,\lambda}(t)=\frac{1}{\lambda}\left(E_{\theta}(t)X_{\theta,\widehat{x_{0}},\bar{u}_{\theta,\lambda}}(t)+h_{\theta}(t,\zeta)\right)\label{eq:OptimalControl}
\end{equation}
where $E_\theta$ and $h_\theta$ are solutions of the Initial Value Problem 
\begin{equation}
\begin{cases}
\dot{E_{\theta}}(t)=C^{T}C-A_{\theta}^{T}E_{\theta}-E_{\theta}A_{\theta}-\frac{1}{\lambda}E_{\theta}^{2},\\
\dot{h_{\theta}}(t,\zeta)=-\alpha_{\theta}(t)h_{\theta}(t,\zeta)-\beta_{\theta}(t,\zeta)\\
\end{cases}
\label{eq:Riccati_ode}
\end{equation}
with $\left(E_{\theta}(0),h_{\theta}(0,\zeta)\right)=\left(Q,0\right)$.  The functions $\alpha_\theta$ and $\beta_\theta$ are defined by 
\[
\left\{ \begin{array}{l}
\alpha_{\theta}(t)=\left(A_{\theta}(t)^{T}+\frac{E_{\theta}(t)}{\lambda}\right)\\
\beta_{\theta}(t,\zeta)=C^{T}\zeta+E_{\theta}r_{\theta}
\end{array}\right.
\] For all $t\in[0,T]$, the matrix $E_{\theta}(t)$ is symmetric, and
the ODE defining the matrix-valued function $t\mapsto E_{\theta}(t)$
is called the Matrix Riccati Differential Equation of the ODE\textup{
(\ref{eq:ControlLinearODEmodel-1}). }

Finally, the Profiled Cost $S$ has the closed form: \textup{
\begin{equation}
\begin{array}{lll}
S\left(\zeta;\theta,\lambda\right)&=  &\int_{0}^{T}\left(  \Vert \zeta(t)\Vert^2  - 2 r_{\theta}(t)^{T}h_{\theta}(t,\zeta)- \frac{1}{\lambda} \Vert h_{\theta}(t,\zeta) \Vert^2 \right)dt\\
&  &  - h_{\theta}(T,\zeta)^{T}E_{\theta}(T)^{-1}h_{\theta}(T,\zeta).
  \label{eq:ExpressionS}
\end{array}
\end{equation}
}
and the final state is estimated by 
 \textup{
\begin{equation}
X_{\theta,\widehat{x_{0}},\bar{u}_{\theta,\lambda}}(T)=-E_{\theta}(T)^{-1}h_{\theta}(T,\zeta).\label{eq:final_state_expression}
\end{equation}
}\end{namedthm}

\begin{rem}
The functions $t\mapsto \left( E(t),h(t)\right) $ are classically called the \emph{adjoint model}. They depend also on $\theta$,
$\lambda$ and $\zeta$ because of their definition via equation (\ref{eq:Riccati_ode}).
Nevertheless, we do not write it systematically for notational brevity.
As mentioned in the theorem, it is possible to compute $X_{\theta,\widehat{x_{0}},\overline{u}_{\theta,\lambda}}$
in a ``closed-loop'' form as we can solve in a preliminary stage
the adjoint model (\ref{eq:Riccati_ode}) that gives the function $E$
and $h$ for all $t\in\left[0,T\right]$. Thanks to equation (\ref{eq:OptimalControl}), 
the closed-form expression of the optimal control  
$\bar{u}_{\theta,\lambda}$ can be plugged into (\ref{eq:ControlLinearODEmodel-1}). We can compute  $X_{\theta,\widehat{x_{0}},\overline{u}_{\theta,\lambda}}$ by solving the following Final Value Problem: 
\begin{equation}
\begin{cases}
\dot{x}(t)=\left(A_{\theta}(t)+\frac{E(t)}{\lambda}\right)x(t)+r_{\theta}(t)+\frac{h(t,\widehat{Y})}{\lambda}\\
x(T)=-E(T)^{-1}h(T,\widehat{Y}). \label{eq:BackwardClosedLoop}
\end{cases}
\end{equation}
The estimate of $\widehat{x_0}=X_{\theta,\widehat{x_{0}},\overline{u}_{\theta,\lambda}}(0)$ of the initial condition is simply the initial value of the Backward ODE (\ref{eq:BackwardClosedLoop}).  
Then by using $X_{\theta,\widehat{x_{0}},\overline{u}_{\theta,\lambda}}$, 
we can compute effectively the control $\bar{u}_\theta$ thanks to (\ref{eq:OptimalControl}). 
\end{rem}

The existence of the criterion $S$ and the fundamental expression (\ref{eq:ExpressionS}) heavily relies on the nonsingularity of the final value of the Riccati solution $E_\theta(T)$. In particular, the final state is estimated by $X_{\theta,\widehat{x_{0}},\bar{u}_{\theta,\lambda}}(T)=-E_{\theta}(T)^{-1}h_{\theta}(T,\zeta)$, and it is then critical to identify the assumptions that could prevent $E_\theta$ to be singular. Our "Theorem and Definition" is legitimate (and proved in the appendix), because the assumption $Q>0$ ensures that $E_\theta(t)$ is nonsingular for all $t$ in $[0,T]$. 
Moreover, the matrix $Q$ can be thought as a kind of prior for helping the state inference. In our basic definition of the cost (\ref{thm:kalman_existence_unicity}), we put a prior on the norm of the initial condition and our regularization penalizes "huge" solutions. Nevertheless, we can have a more refined prior and use a preliminary guess $\mu \in \mathbb{R}^d$. The modification of  the criterion is straightforward by setting 
\[ C_\mu \left(\zeta;u,x_{0},\theta,\lambda\right)=(x_{0}-\mu)^{T}Q (x_{0}-\mu)+ \left\Vert \zeta(t)-CX_{\theta,x_{0},u}(t)\right\Vert _{L^2}^{2}+\lambda \left\Vert u(t)\right\Vert _{L^2}^{2}. \]
By re-parameterizing the initial conditions with $y_0=x_0 - \mu$ and exploiting the relation $X_{\theta,x_0-\mu,u}(t)=X_{\theta,x_0-\mu,u}(t) - \Phi_\theta(t,0)\mu$ (consequence of the linearity of the ODE) , we get that  \[ \inf_{x_0} \inf_u C_\mu \left(\zeta;x_{0},u,\theta,\lambda\right) = S\left(\zeta - C\Phi_\theta(\cdot,0)\mu ;\theta,\lambda\right).
 \]
At the opposite, it might be inappropriate in some circumstances to impose such kind of information for the initial condition. This can be the case if the number of observations tends to infinity and  $\hat{Y}$ becomes quite close to the truth. Another situation is when the initial conditions of the unobserved part are largely unknown.  
Hence, we extend our estimator to the case $Q=0$, that corresponds also to our framework for studying the asymptotics of $\hat{\theta}^K$ . In order to derive relevant and tractable conditions for ensuring the existence of $S$, we need to ensure that only one trajectory, with a unique initial condition (or final condition), is the global minimum of $C(\zeta;u,x_0,\theta,\lambda)$. The nonsingularity of $E_\theta(t)$ is in fact related to the concept of observability in control theory. In the next proposition, we will pave the way to the assumptions on $C$ and the vector field $A_\theta$ that can guarantee the general existence of our method. 

\begin{prop} \label{prop:invertibility_general_criteria} 
For a given parameter $\theta\in\Theta$ and observation matrix $C$,  the properties 1 and 2  are equivalent:  
\begin{enumerate}
\item The system outputs $Y(t)=C\Phi(t,0)x_0$ satisfy 
\begin{equation}
\int_{0}^{T}\left\Vert C\Phi_\theta(t,0) x_{0,1} - C\Phi_\theta(t,0) x_{0,2}  \right\Vert_{2}dt=0\Longrightarrow x_{0,1}=x_{0,2} \label{eq:first_equivalent_invertibility_criteria}
\end{equation}
\item The (final) observability matrix $O_\theta(T)$ is nonsingular
\begin{equation}
O_\theta(T)=\int_{0}^{T}\Phi_{\theta}(t,0)^{\top}C^{\top}C\Phi_{\theta}(t,0)dt\label{eq:matrix_criteria_invertibility}
\end{equation}
\end{enumerate}
If one of the properties is satisfied, then $E_{\theta}(T)$ is nonsingular and $S$ is defined for $Q\geq0$.
\end{prop}
An important feature of that proposition is that the criterion does not depend on $r_{\theta}$. Moreover, if $C$ is full rank, the matrix $E_{\theta}(T)$
is always nonsingular for all $\theta$ in $\Theta$. 
The criterion 1 means that for a given $\theta$, any solution
$X_{\theta,x_{1,0}}$ and $X_{\theta,x_{2,0}}$ of (\ref{eq:LinearODEmodel})
can be distinguished by their partial observation $Y_{\theta}^{i}(t):=CX_{\theta,x_{i,0}}(t), i=1,2$. The matrix $C$
"gives" enough information about the system so that the observed
part is sufficient to uniquely characterize the whole system's state.

The next section is dedicated to the derivation of the regularity
properties of $S$. Thanks to the different possible expressions for the criterion  $S$, we can show the smoothness in $\zeta$
and $\theta$, and compute directly the needed derivatives.

\section{\label{sec:Estimator_properties}Consistency of the Deterministic Kalman Filter Estimator}

\subsection{Properties of the criterion $S(\widehat{Y};\theta,\lambda)$}

We have a tractable expression of the cost function $S(\widehat{Y};\theta,\lambda)$
for a given $\theta$, but we still need to derive the properties of
$\theta \mapsto S(\widehat{Y};\theta,\lambda)$ and $\theta\rightarrow S(Y^{*};\theta,\lambda)$
on $\Theta$, and shows some convergence properties.
First of all, we need to ensure the existence of $S(\widehat{Y};\theta,\lambda)$; this is the case if the non-parametric estimator $\widehat{Y}$  belongs to $L^{\infty}(\left[0,\, T\right],\mathbb{R}^{d'})$ (more explanations are given in appendix A). 
We show that for all   $Y$ in $L^{\infty}(\left[0,\, T\right],\mathbb{R}^{d'})$, the function 
$\theta \mapsto S(Y;\theta,\lambda)$ is well defined and $C^{1}$
on $\Theta$, under some regularity and identifiability assumptions, detailed below:
\begin{description}
\item[C1:]    $\Theta$ is a compact subset of $\mathbb{R}^{p}$ and $\theta^{*}$ is in the interior $\mathring{\Theta}$,
\item[C2a:]  $Q=0$  and for all $\theta$ in $\Theta$, $O_\theta(T)$ is nonsingular,
\item[C2b:] The model is identifiable at $(\theta^{*},x_{0}^{*})$ i.e 
\[
\forall\left(\theta,x_{0}\right)\in\Theta\times\mathcal{X}\,;\, CX_{\theta,x_{0}}=CX_{\theta^{*},x_{0}^{*}}\Longrightarrow (\theta, x_0)= (\theta^{*}, x_{0}^{*}),
\]
\item[C3:] $\forall\left(t,\theta\right)\in\left[0\,,\, T\right]\times\Theta,\:(t,\theta)\rightarrow A_{\theta}(t)$
and $(t,\theta)\rightarrow r_{\theta}(t)$ are continuous,
\item[C4:] $\forall\left(t,\theta\right)\in\left[0,\, T\right]\times\Theta,\,\left(t,\theta\right)\longmapsto\frac{\partial A_{\theta}}{\partial\theta}$
and $\left(t,\theta\right)\longmapsto\frac{\partial r_{\theta}}{\partial\theta}$
are continuous.
\end{description}
Condition 2 is about identifiability condition: condition 2a is needed for the existence of the criterion $S$, and is related to the identifiability of the initial condition. But C2a is not sufficient, and we need Condition 2b for structural identifiability, based on the joint identifiability at $\left(\theta^{*},x_{0}^{*}\right)$. We require that the observed
output $CX_{\theta^{*},x_{0}^{*}}$ can be generated on by the couple
$\left(\theta^{*},x_{0}^{*}\right)$. 
The identifiability problem of systems can be difficult (more than observability). For linear
system, several approaches can be used, such as Laplace Transform \cite{Bellman1970}, or Power Expansions \cite{Pohjanpalo1978}, see \cite{MiaoXiaPerelsonWu-siam2011} for a review. 
So far, most of existing methods are poorly used because they rely on (heavy) formal computations, which
limit their interest to low dimensional system. Nonetheless, progress in automatic formal computation has promoted new methods based on
differential algebra and the Ritt's algorithm, that improves identifiability checking, \cite{Hubert1999,Lyzell2011,Vidal2000}. 

According to the context, the norm $\left\Vert \right\Vert _{2}$ will
denote the Euclidean norm in $\mathbb{R}^{d}$, $\left\Vert X\right\Vert _{2}=\sqrt{\sum_{i=1}^{d}X_{i}^{2}}$
or the Frobenius matrix norm $\left\Vert A\right\Vert _{2}=\sqrt{\sum_{i,j}\left|a_{i,j}\right|^{2}}$. We use the functional norm in $L^{2}\left(\left[0,\, T\right],\mathbb{R}^{d'}\right)$
defined by: $\left\Vert f\right\Vert _{L^{2}}=\sqrt{\int_{0}^{T}\left\Vert f(t)\right\Vert _{2}^{2}dt}$.
Continuity and differentiability have to be understood according to
these norms.
\begin{prop}
\label{prop:existence_A_E_h}Under conditions 1, 2a and 3 we have: 

\[
\overline{A}=\sup_{\theta\in\Theta}\left\Vert A_{\theta}\right\Vert _{L^{2}}<+\infty
\]
\[
\overline{r}=\sup_{\theta\in\Theta}\left\Vert r_{\theta}\right\Vert _{L^{2}}<+\infty
\]
 
\[
\overline{X}=\sup_{\theta\in\Theta}\left\Vert X_{\theta}\right\Vert _{L^{2}}<+\infty
\]
 
\[
\bar{E}=\sup_{\theta\in\Theta}\left\Vert E_{\theta}\right\Vert _{L^{2}}<+\infty
\]
 and 
\[
\forall Y\in L^{\infty}(\left[0,\, T\right],\mathbb{R}^{d'}),\,\bar{h_{Y}} = \sup_{\theta\in\Theta}\left\Vert h_{\theta}(.,Y)\right\Vert _{L^{2}}<+\infty
\] Hence, for all  $Y$ in $L^{\infty}(\left[0,\, T\right],\mathbb{R}^{d'})$, 
the map $\theta\longmapsto S(Y;\theta,\lambda)$ is well defined on
$\Theta$ (i.e $\sup_{\theta\in\Theta}S(Y;\theta,\lambda)<+\infty$) 
\end{prop}
We have shown that for all $Y$ in $L^{\infty}\left(\left[0\, T\right],\mathbb{R}^{d'}\right)$, the maps
$\theta \mapsto S(Y;\theta,\lambda)$ is well defined and so are
$\theta \mapsto S(\widehat{Y};\theta,\lambda)$ and $\theta\mapsto S(Y^{*};\theta,\lambda)$
as long as the non-parametric estimator $\widehat{Y}$ are well-defined on $\left[0,\, T\right]$.
\begin{prop}
\label{prop:c1_caracterisation_S} Under conditions 1, 2, 3 
\[
\forall Y\in L^{\infty}(\left[0,\, T\right],\mathbb{R}^{d'}),\,\theta\longmapsto S(Y;\theta,\lambda)
\]
 is continuous on $\Theta$. Under conditions 1, 2a, 3, 4 it is $C^{1}$
on $\Theta$.
\end{prop}
In proposition \ref{prop:existence_A_E_h} we have shown that our
criteria $\theta\longmapsto S(Y;\theta,\lambda)$ is well defined
( i.e $0\leq S(Y;\theta,\lambda)<+\infty$) and here we have demonstrated
(using regularity assumptions on the model) that our finite and asymptotic
criteria are continuous or even $C^{1}$ on $\Theta$. Theses regularity
properties justify the use of classical optimization method to retrieve
the minimum of $S(\widehat{Y};.,\lambda)$.

\subsection{\label{sec:Consistency}Consistency}

We show the consistency of the parameter estimator $\hat{\theta}^K $ when the model is well-specified. As already mentioned, we have defined an $M$-estimator, and we can proove the consistency (see \cite{Vaart1998}), by showing 
\begin{enumerate}
\item the uniform convergence of $S(\widehat{Y};\theta,\lambda)$ to  $S(Y^{*};\theta,\lambda)$
on $\Theta$, 
\item $\theta^{*}$ is the unique global minimum of the asymptotic criterion $S(Y^{*};\theta,\lambda)$ on $\Theta$.
\end{enumerate}
The second point is assessed in proposition \ref{prop:asym_identifiability}, and it is related to the structure identifiability of the model provided by condition 2b.
\begin{prop}
\label{prop:asym_identifiability}Under conditions 1, 2a, 2b,  we have: 
\[
S(Y^{*};\theta,\lambda)=0\Longleftrightarrow\theta=\theta^{*}
\]
\end{prop}
Point 1 is proved by studying the regularity of the map $(\zeta,\theta) \mapsto S(\zeta; \theta,\lambda)$ and by obtaining appropriate controls of the variations by $\hat{Y}-Y^*$, see \emph{Supplementary Materials}. Theorem (\ref{thm:consistency})  can be claimed with some generality on the nonparametric proxy $\hat{Y}$. 
\begin{thm}
\label{thm:consistency} Under conditions 1, 2a, 2b, 3 and if 
$\widehat{Y}$ is consistent in probability in $L2$, then $\widehat{\theta}^{K}\overset{P}{\rightarrow}\theta^{*}$. \end{thm}

\section{\label{sec:Asymptotic_Normality}Asymptotics of $\widehat{\theta}^K$}

The aim of this section is to derive the rate of convergence and asymptotic law of  $\widehat{\theta}^K$. For this reason, we need more precise assumptions on $\hat{Y}$. The way we proceed is based on the plug-in properties of nonparametric estimates, when the functionals of interest are relatively smooth. In the case of series expansion, these properties are well understood \cite{Newey1997,BickelRitov2003}. We focus here on regression splines, as they are well-used in practice and relatively simple to study, although more refined nonparametric estimators can be used in the same context, such as Penalized Splines. We assume that $\hat{Y}$ has a B-Spline expansion 
\[ \widehat{Y}(t)=\sum_{k=1}^{K}\beta_{kK}p_{kK}(t)=\beta_{K}^{T}p_{K}(t)\]
where $\beta_{K}$ is computed by linear least-squares, and the dimension $K$ increases with $n$.
We introduce then additional regularity conditions on the ODE model, and on the distribution of observations:
\begin{description}
\item[C5:] $\forall\left(t,\theta\right)\in\left[0,\, T\right]\times\Theta,\,\left(t,\theta\right)\longmapsto\frac{\partial^{2}A_{\theta}}{\partial\theta^{T}\partial\theta}$
and $\left(t,\theta\right)\longmapsto\frac{\partial r_{\theta}^{2}}{\partial\theta^{T}\partial\theta}$
are continuous, 
\item[C6:] $\frac{\partial^{2}S(Y^{*};\theta^{*},\lambda)}{\partial\theta^{T}\partial Y}$
is nonsingular,
\item[C7:] The observations $\left(t_{i},Y_{i}\right)$ are i.i.d with $Var(Y_{i}\mid t_{i})=\sigma I_{d'}$
with $\sigma<+\infty$,
\item[C8:] The observation times $t_{i}$ are uniformly distributed on $\left[0\,,\, T\right]$,
\item[C9:]  It exists $s\geq1$ such that $t\longmapsto A_{\theta^{*}}(t)$, $t\longmapsto r_{\theta^{*}}(t)$  are $C^{s-1}\left(\left[0\,,\, T\right])  \mathbb{R}^{d}\right)$ and $\sqrt{n}K^{-s}\longrightarrow0, \frac{K^{4}}{n}\longrightarrow0$
\item[C10:] The meshsize $\max_{i}\left|\tau_{i+1,K+1}-\tau_{i,K}\right|\longrightarrow0$ when$ K\longrightarrow\infty$
\end{description}
The proofs of the rate and asymptotic normality are somewhat technical, and they are relegated in the \emph{Supplementary Materials}. We obtain a  parametric convergence rate, and  the asymptotic
normality, by using two facts: 
\begin{enumerate}
\item $\widehat{\theta}^{K}-\theta^{*}$
behaves like the  difference $\Gamma(\widehat{Y})-\Gamma(Y^{*})$, where $\Gamma$ is a linear functional, 
\item if $\Gamma$ is smooth enough, $\Gamma( \widehat{Y}-Y^{*} )$ is asymptotically
normal in the case of regression splines.
\end{enumerate}
Conditions C5 and C6 ensures the sufficiency of second order optimality
conditions for the criteria $S$. Conditions C7 to C10 are sufficient for
the consistency of $\widehat{Y}$, as well as for 
the consistency and the asymptotic normality of the plug-in estimators of linear functionals. 

\begin{thm}
\label{thm:asymptotic} If conditions C1-C10 are satisfied, then $\widehat{\theta}^{K}-\theta^{*}= O_{P}(n^{-1/2})$ and  $\widehat{\theta}$
is asymptotically normal.
\end{thm}

\section{State Estimation \label{sec:StateEstimation}}

Once the unknown model has been estimated with $\hat{\theta}^{K}$,
we focus on the problem of state estimation. From the definition (\ref{thm:kalman_existence_unicity}),
the criterion $S$ is built with an estimation of the
state based on the solution of the pertubed ODE $X_{\hat{\theta},\widehat{x_{0}},\bar{u}}$.
The estimate of the initial condition $\widehat{x_{0}}$ is
derived from a Final Value Problem with the final state $X_{\hat{\theta},\widehat{x_{0}},\bar{u}}(T)=-\left[E_{\hat{\theta}}(T)\right]^{-1}h_{\hat{\theta}}(T,\hat{Y})$.\\
The state estimate $t\mapsto X_{\hat{\theta},\widehat{x_{0}},\bar{u}}(t)$ that we have used, 
is different from the state estimation classically done when using the Deterministic
Kalman Filter. The classical DKF state estimate is $\hat{X}^{DKF}(t)=-\left[E_{\hat{\theta}}(t)\right]^{-1}h_{\hat{\theta}}(t,\hat{Y}_{\vert[0,t]})$,
and it corresponds to the best estimate of $X^{*}(t)$ computed from
the available information at time $t$, $\hat{Y}_{\vert[0,t]}=\left\{ \hat{Y}(s),s\in\left[0,t\right]\right\} $.
Whereas the estimate can be very bad at the beginning for small $t$,
the quality of $\hat{X}^{DKF}(t)$ improves as we get more data. A
remarkable feature of the filter $t\mapsto\hat{X}^{DKF}(t)$ is that
it can be computed recursively with an Ordinary Differential Equation
$\dot{X}(t)=A_{\theta}X(t)+r_{\theta}(t)+L_{\theta,\lambda}(t)\left(CX(t)-\hat{Y}(t)\right)$.
The matrix $L_{\theta,\lambda}$ is the continuous counterpart of
the classical Kalman Gain Matrix, derived from the Filtering Riccati
Differential Equation, see page 313 in \cite{Sontag1998}. In that
recursive form, the Deterministic Kalman Filter is somehow similar
to the Kalman-Bucy Filter, which is the continuous version of
the usual Kalman Filter. Nevertheless, there is a
huge difference in the assumptions because
the Kalman-Bucy Filter assumes that $X(t)$ is a Stochastic Differential
Equation, driven by a Brownian Motion $W(t)$. This means that  the
deterministic perturbation $u(t)$ is replaced by a random
pertubation $\sigma dW(t)$. The state estimate is then different
from the one we consider as it can be shown that the filter is the
solution of a stochastic differential equation driven by the stochastic
process $\left(CX(t)-\hat{Y}(t)\right)$, see for instance \cite{Bensoussan2004}.
The state estimate $t\mapsto X_{\hat{\theta},\widehat{x_{0}},\bar{u}}(t)$
is solution of the pertubed ODE, with the control $\bar{u}$ computed
from all the data $\left\{ \hat{Y}(s),s\in\left[0,T\right]\right\} $:
hence, our state estimation is based on Kalman Smoothing and not on
Filtering, as we have a backward integration step. In the rest of
that section, we show that the estimator $X_{\hat{\theta},\widehat{x_{0}},\bar{u}}(t)$
is also a consistent estimator of the state $X^{*}(t)$. In order
to do that, we show first that $\widehat{X}(T)=-E_{\widehat{\theta}}(T)^{-1}h_{\widehat{\theta}}(T,\widehat{Y})$
is a consistent estimator of the final state.

\subsection{Final state estimation }

In a way, the consistency of the final state estimator is a rather
obvious conclusion. The Deterministic Kalman Filter is initially designed
for getting the best possible estimate of the final state, starting
from any initial condition $x_{0}$. It is then normal that we have
a good estimator of $X^{*}(T)$ when $\hat{Y}$ is close to $Y^{*}$
and $\hat{\theta}^{K}$ is close to $\theta^{*}$. 
\begin{prop}
\label{prop:final_state_estimator_consistency} We assume that conditions
C1-C4 are satisfied and that $\widehat{Y}$ is a consistent estimator of $Y^{*}$.
Then, the final state estimator $\widehat{X}(T)=-E_{\widehat{\theta}}(T)^{-1}h_{\widehat{\theta}}(T,\widehat{Y})$
converges in probability to $X^{*}(T)$.\textup{ }\end{prop}
\begin{proof}
We show first that the true final state value is reached for $Y=Y^{*}$
and $\theta=\theta^{*}$ i.e $X^{*}(T)=-E_{\theta^{*}}(T)^{-1}h_{\theta^{*}}(T,Y^{*})$.
We recall that 
\[
S\left(Y;\theta,\lambda\right)=\inf_{x_{0}\in\mathbb{R}^{d}}\left(\inf_{u\in L^{2}}C(Y;x_{0},u,\theta,\lambda)\right),
\]
and that $S\left(Y^{*};\theta^{*},\lambda\right)=0$. The identifiability
condition 2b implies that the reconstructed state is the exact one.
In our case, the minimum is reached when the optimal control $\overline{u}$
is equal to $0$, i.e. 
\[
\bar{u}_{\theta^{*},\lambda}(t)=\frac{1}{\lambda}\left(E_{\theta^{*}}(t)X^{*}(t)+h_{\theta^{*}}(t,Y^{*})\right)
\]
which implies that $X^{*}(T)=-E_{\theta^{*}}(T)^{-1}h_{\theta^{*}}(T,Y^{*})$
($E_{\theta^{*}}(T)$ is nonsingular). We can decompose the difference
$\widehat{X}(T)-X^{*}(T)$:
\[
\begin{array}{lll}
\widehat{X}(T)-X^{*}(T) & = & E_{\theta^{*}}(T)^{-1}h_{\theta^{*}}(T,Y^{*})-E_{\widehat{\theta}}(T)^{-1}h_{\widehat{\theta}}(T,\widehat{Y})\\
 & = & E_{\theta^{*}}(T)^{-1}\left(h_{\theta^{*}}(T,Y^{*})-h_{\widehat{\theta}}(T,\widehat{Y})\right)\\
 & + & \left(E_{\theta^{*}}(T)^{-1}-E_{\widehat{\theta}}(T)^{-1}\right)h_{\widehat{\theta}}(T,\widehat{Y})
\end{array}
\]
The convergence will come from the consistency of $h_{\widehat{\theta}}(T,\widehat{Y})$
and $E_{\widehat{\theta}}(T)^{-1}$: 

\[
\begin{array}{lll}
\left|\widehat{X}(T)-X^{*}(T)\right| & \leq & \sqrt{d}\left\Vert E_{\theta^{*}}(T)^{-1}\right\Vert _{2}\left\Vert h_{\theta^{*}}(T,Y^{*})-h_{\widehat{\theta}}(T,\widehat{Y})\right\Vert _{2}\\
 & + & \sqrt{d}\left\Vert h_{\widehat{\theta}}(T,\widehat{Y})\right\Vert _{2}\left\Vert E_{\theta^{*}}(T)^{-1}-E_{\widehat{\theta}}(T)^{-1}\right\Vert _{2}
\end{array}
\]
The two right-hand side terms can be controlled easily by the $\hat{Y}-Y^{*}$
and $\hat{\theta}-\theta^{*}$, as it is shown in Lemma \ref{lem:h_discrepancy_lambda}
and \ref{lem:E_inverse_discrepancy_wrt_lambda}. We end up with the
following inequalities:
\[
\begin{array}{lll}
\left\Vert h_{\theta}(t,Y)-h_{\theta^{'}}(t,Y')\right\Vert _{2} & \leq & K_{6}e^{L_{1}\frac{\overline{E_{\lambda}}}{\lambda}}\left\Vert Y-Y'\right\Vert _{L^{2}}\\
 & + & \left(K_{7}+K_{8}\overline{E_{\lambda}}\right)\left(K_{4}+\frac{K_{5}}{\lambda}\overline{E_{\lambda}}e^{L_{1}\frac{\overline{E_{\lambda}}}{\lambda}}\right)e^{2L_{1}\frac{\overline{E_{\lambda}}}{\lambda}}\left\Vert \theta-\theta^{'}\right\Vert \\
 & + & \left(K_{9}e^{2L_{1}\frac{\overline{E_{\lambda}}}{\lambda}}+K_{10}e^{L_{1}\frac{\overline{E_{\lambda}}}{\lambda}}\right)\overline{E_{\lambda}}\left\Vert \theta-\theta^{'}\right\Vert \\
\left\Vert E_{\theta}^{-1}(T)-E_{\theta'}^{-1}(T)\right\Vert _{2} & \leq & \left(\frac{K_{12}}{\lambda}+K_{11}\right)e^{K_{13}+\frac{K_{14}}{\lambda}}\left\Vert \theta-\theta'\right\Vert \\
\left\Vert E_{\theta}^{-1}(T)\right\Vert  & \leq & \frac{K_{15}}{\lambda}\\
\left\Vert h_{\theta}(T,Y)\right\Vert _{2} & \leq & \sqrt{T}d^{2}\left\Vert C\right\Vert _{2}e^{\sqrt{d}\left(\overline{A}+\frac{\overline{E_{\lambda}}}{\lambda}\right)T}\left\Vert Y\right\Vert _{L^{2}}+\sqrt{d}\overline{E_{\lambda}}\overline{r_{\theta}}
\end{array}
\]
Under our conditions, we have $\left(\hat{\theta},\widehat{Y}\right)\longrightarrow\left(\theta^{*},Y^{*}\right)$,
which implies that $\widehat{X}(T)$ converges also in probability.
\end{proof}
By plug-in principle, we can also derive the asymptotic normality
and the rate of $\hat{X}(T)$ as described in the next proposition. 
\begin{prop}
\label{prop:final_state_estimator_asymptotic_normality} Under conditions
C1-C10, the final state estimator $\widehat{X}(T)$ is asymptotically
normal and 
\[
\widehat{X}(T)-X^{*}(T)=O_{P}(n^{-1/2})
\]
\end{prop}
\begin{proof}
We have the following decomposition: 
\[
\begin{array}{lll}
\widehat{X}(T)-X^{*}(T) & = & E_{\theta^{*}}(T)^{-1}h_{\theta^{*}}(T,Y^{*})-E_{\widehat{\theta}}(T)^{-1}h_{\widehat{\theta}}(T,\widehat{Y})\\
 & = & E_{\theta^{*}}(T)^{-1}\left(h_{\theta^{*}}(T,Y^{*})-h_{\theta^{*}}(T,\widehat{Y})\right)\\
 & + & E_{\theta^{*}}(T)^{-1}\left(h_{\theta^{*}}(T,\widehat{Y})-h_{\widehat{\theta}}(T,\widehat{Y})\right)\\
 & + & \left(E_{\theta^{*}}(T)^{-1}-E_{\widehat{\theta}}(T)^{-1}\right)h_{\widehat{\theta}}(T,\widehat{Y})
\end{array}
\]
According to Theorem 7 in \cite{Newey1997} $\widehat{Y}$ is a consistent
estimator of $Y^{*}$ hence using proposition \ref{prop:final_state_estimator_consistency}
and continuous mapping theorem we have: 
\[
\widehat{X}(T)-X^{*}(T)=E_{\theta^{*}}(T)^{-1}\left(h_{\theta^{*}}(T,Y^{*})-h_{\theta^{*}}(T,\widehat{Y})\right)+o_{p}(1)
\]
Using the linear representation for $h_{\theta^{*}}$ we obtain:

\[
E_{\theta^{*}}(T)^{-1}\left(h_{\theta^{*}}(T,Y^{*})-h_{\theta^{*}}(T,\widehat{Y})\right)=-E_{\theta^{*}}(T)^{-1}\int_{0}^{T}R_{\theta^{*}}(T,s)C^{T}\left(\widehat{Y}(s)-Y^{*}(s)\right)ds
\]
We define 
\[
H(t,\theta).Y=E_{\theta}(T)^{-1}R_{\theta}(T,t)C^{T}Y(t)
\]
the linear form such that 
\[
-E_{\theta^{*}}(T)^{-1}\left(h_{\theta^{*}}(T,Y^{*})-h_{\theta^{*}}(T,\widehat{Y})\right)=\int_{0}^{T}\left(H(s,\theta^{*}).Y^{*}-H(s,\theta^{*}).\widehat{Y}\right)ds
\]
As for the normality of $\hat{\theta}^{K}$, we can use theorem 9
in \cite{Newey1997} in order the obtain the asymptotic normality
of $\int_{0}^{T}\left(H(s,\theta^{*}).Y^{*}-H(s,\theta^{*}).\widehat{Y}\right)ds$
with $\sqrt{n}-$rate. 
\end{proof}

\subsection{Estimation of the states on $\left[0,T\right]$ and influence of
$\lambda$\label{sub:StatesEstimation}}

We can estimate the trajectory $X^{*}(t)$ with the smoothed trajectory
$t\mapsto X_{\hat{\theta},\hat{x}_{0},\bar{u}}(t)$ or with the exact
model $t\mapsto X_{\hat{\theta},\hat{x}_{0},0}$, without the perturbation
$\bar{u}$. We need then to have a better understanding of the quality
of these two estimates, and in particular of the relevancy of $\hat{x}_{0}$,
defined as the unknown initial condition of the Final Value Problem
(\ref{eq:BackwardClosedLoop}). We have profiled the initial condition
in the definition of $S$, in order to separate the estimation of
$\theta^{*}$ from the estimation of the initial condition. Nevertheless,
the estimation of the states is a by-product of the parameter estimation,
and the remaining point in our analysis is to ensure that $\hat{x}_{0}$
is really a good estimator for $x_{0}^{*}$. This is the case, and
we will show more generally that $X_{\hat{\theta},\hat{x}_{0},\bar{u}}$
is a good estimator of the trajectory $X^{*}$. Quite remarkably,
the consistency of $X_{\hat{\theta},\hat{x}_{0},\bar{u}}$ is the
first result that relies on a assumption on the hyperparameter $\lambda$.
This is due to the fact that $\bar{u}$ is a perturbation computed
for tracking $\hat{Y}$, while taking into account the model uncertainty
estimated by $\hat{\theta}$ instead of $\theta^{*}$. The convergence
of $\hat{Y}$ to $Y^{*}$ and the identifiability conditions 2a and
2b (plus regularity conditions) are sufficient to ensure the convergence
of $\hat{\theta}$ to $\theta^{*}$, without particular assumptions
on $\lambda$. This is possible because the true model $X_{\theta,x_{0},0}$
is included into the perturbed model $X_{\theta,x_{0},u}$. 

If $\lambda$ is not big enough, the size of the perturbation $\left\Vert \bar{u}\right\Vert _{L^{2}}^{2}$
is not highly constrained in the cost function $S$, and we can have
overfitting: the estimator $CX_{\hat{\theta},\hat{x}_{0},\bar{u}}$
can be quite close to $\hat{Y}$ with a ``big'' $\bar{u}$ that
makes $X_{\hat{\theta},\hat{x}_{0},\bar{u}}$ far from of $X^{*}$.
This problem can be even more important, if we have errors on $\hat{\theta}$,
because $\bar{u}$ will have to compensate the errors in the parameter
estimation. In that case, we cannot guarantee to have a consistent
estimate for $x_{0}$, if we don't have $\lambda\longrightarrow\infty$.
Indeed, the trajectory $X_{\hat{\theta},\hat{x}_{0},\bar{u}}$ is
the solution to the pertubed initial value problem

\begin{equation}
\begin{cases}
\dot{x}(t)= & \left\{ A_{\hat{\theta}}(t)+\frac{1}{\lambda}E_{\hat{\theta},\lambda}(t)\right\} x(t)+r_{\hat{\theta}}(t)+\frac{1}{\lambda}h(t,\widehat{Y})\\
x(T)= & -E_{\hat{\theta},\lambda}(T)^{-1}h(T,\widehat{Y})
\end{cases}\label{eq:EstimatedPertubedODE}
\end{equation}
Because of the convergence of $\left(\hat{\theta},\widehat{Y}\right)\longrightarrow\left(\theta^{*},Y^{*}\right)$,
we can ensure the convergence to the right trajectory if we control
$\lambda$. 
\begin{prop}
Under conditions C1-C10 and if $\lambda_{n}\longrightarrow\infty$,
then 
\[
X_{\hat{\theta},\hat{x}_{0},\bar{u}}(t)\longrightarrow X^{*}(t)
\]
 for all $t\in[0,T]$. Moreover, 
\[
\widehat{x}_{0}-x_{0}^{*}=O_{P}(n^{-1/2}).
\]
\end{prop}
\begin{proof}
We first need to show that for all $\theta\in\Theta$, and $\zeta$,
the functions $E_{\theta,\lambda}$ and $h_{\theta,\lambda}(\cdot,\zeta)$
are bounded (as they converge) when $\lambda\longrightarrow\infty$.
As $E_{\theta,\lambda}$ is solution of the matrix equation $\dot{E}=C^{\top}C-A^{\top}E-EA-\frac{1}{\lambda}E^{2}$
that depends smoothly in $\lambda^{-1}$; hence $E_{\theta,\lambda}\longrightarrow E_{\theta,\infty}$
defined as the solution of the linear matrix ODE $\dot{E}=C^{\top}C-A^{\top}E-EA$
(with $E(0)=0$). Moreover, $h_{\theta,\lambda}(\cdot,\zeta)$ is
solution of the linear ODE $\dot{h}=-\alpha_{\theta,\lambda}h-\beta_{\theta,\lambda}(\cdot,\zeta)$
with $\alpha_{\theta,\lambda}\longrightarrow A_{\theta}^{\top}$ and
$\beta_{\theta,\lambda}(\cdot,\zeta)\longrightarrow\beta_{\theta,\infty}=C^{\top}\zeta+E_{\theta,\infty}r_{\theta}$.
As the dependency in $\lambda^{-1}$ is smooth, the solution $h_{\theta,\lambda}(\cdot,\zeta)$
converges to $h_{\theta,\infty}$, solution of $\dot{h}=-\alpha_{\theta,\infty}h-\beta_{\theta,\infty}(\cdot,\zeta)$.
Additionally, the dependency in $\left(\lambda,\zeta\right)$ is smooth
on $\mathbb{R}^{+}\times L^{2}$ and $h_{\theta,\lambda}(\cdot,\zeta)$
converges to $h_{\theta,\infty}(\cdot,Y^{*})$ as $(\lambda,\zeta)$
tends to $(\infty,Y^{*})$. This means that if $(\theta,Y)$ converges
to $(\theta^{*},Y^{*})$ as $\lambda\longrightarrow\infty$, then
$X_{\theta,x_{0},\lambda}$ converges to the solution of the final
value problem 
\begin{equation}
\begin{cases}
\dot{x}(t)= & A_{\theta^{*}}(t)x(t)+r_{\theta^{*}}(t)\\
x(T)= & X_{\theta^{*},x_{0}^{*},0}(T)
\end{cases}\label{eq:EstimatedPertubedODE-1}
\end{equation}
as $\lambda^{-1}E_{\theta,\lambda}$ and $\lambda^{-1}h_{\theta,\lambda}(t,Y)$
tends to zero when $\lambda\longrightarrow\infty$, and if $x(T)\longrightarrow X^{*}(T)$.
Because of the uniqueness of the solutions to Initial or Final Value
Problem, we have $X_{\theta,x_{0},\lambda}\longrightarrow X^{*}$
for all $t\in\left[0,T\right]$. 

In proposition \ref{prop:final_state_estimator_consistency}, under
conditions C1-C8, we have shown that $(\hat{\theta}_{n}^{K},\hat{Y}_{n})$
converges in probability to $(\theta^{*},Y^{*})$ for all $\lambda$
on $\Theta\times L^{2}$. By the continuous mapping theorem applied
to $(\hat{\theta}_{n}^{K},\hat{Y}_{n},\lambda_{n})$, with $\lambda_{n}\longrightarrow\infty$,
we have $X_{\hat{\theta},\hat{x}_{0},\bar{u}_{\lambda_{n}}}(t)\longrightarrow X^{*}(t)$
for all $t\in\left[0,T\right]$ in probability. In particular, we
obtain the convergence of $\hat{x}_{0}$ to $x_{0}^{*}$. 

The asymptotic normality and root-n rate of $\hat{x}_{0}$ comes from
the asymptotic normality and rates of $\hat{\theta}$ and $\hat{X}(T)$.
If $\psi_{\theta,\lambda}(t,0)$ is the resolvant of the ODE $\dot{x}(t)=\left\{ A_{\theta}(t)+\frac{1}{\lambda}E_{\theta,\lambda}(t)\right\} x(t)$,
we have a closed form for the smoother 
\[
X_{\hat{\theta},\hat{x}_{0},\bar{u}_{\lambda_{n}}}(t)=\psi_{\hat{\theta},\lambda}(t,0)\hat{x}_{0}+\psi_{\hat{\theta},\lambda}\int_{0}^{t}\left[\psi_{\hat{\theta},\lambda}(s,0)\right]^{-1}\left\{ r_{\hat{\theta}}(s)+\frac{1}{\lambda}h_{\hat{\theta},\lambda_{n}}(s,\hat{Y})\right\} ds
\]
When we evaluate at $t=T$, we obtain the following formula for the
initial state $\hat{x}_{0}=\left[\psi_{\hat{\theta},\lambda}(T,0)\right]^{-1}\hat{X}(T)-\int_{0}^{T}\left[\psi_{\hat{\theta},\lambda}(s,0)\right]^{-1}\left\{ r_{\hat{\theta}}(s)+\frac{1}{\lambda_{n}}h_{\hat{\theta},\lambda_{n}}(s,\hat{Y})\right\} ds$.
Hence, $\hat{x}_{0}$ is a smooth transformation of $\left(\hat{\theta},\hat{X}(T)\right)$,
and we can conclude by the parametric delta-method. 
\end{proof}

\subsection{Choice of $\lambda$ and cross-validation}

Our theoretical analysis shows that when $n$ tends
to infinity, we have a family of good estimates $(\hat{\theta}_{\lambda_{n}}^{K},\widehat{x_{0,\lambda_{n}}})$,
with $\lambda_{n}\longrightarrow\infty$. The remaining question is
 to define an appropriate selection procedure for $\lambda$,
that could be used in practice with a finite number of observations
$(y_{1},\dots,y_{n})$. A straightforward way of selecting $\lambda$
is to use a cross-validation selection procedure. Indeed, our criterion
$S\left(\hat{Y};\theta,\lambda\right)$ is based on a balance between
data fidelity and model fidelity, and a rough analysis shows that
when $\lambda\longrightarrow0$, we can select any $u$ in order to
interpolate $\widehat{Y}$ and $\theta$ has almost no influence on
$S\left(\hat{Y};\theta,\lambda\right)$. Whereas when $\lambda\longrightarrow\infty$,
the optimal perturbation $\bar{u}\longrightarrow0$, and we get a
NLS-like criterion where the observations $Y_{i}$'s are replaced
by the proxy $\widehat{Y}$. 

A good hyperparameter $\lambda_{n}$
should give a good estimate of the states $X^{*}(t)$ (and of the
output $Y^{*}(t)$), even if we are only interested in parameter estimation.
Anyway, if we want to use the minimization of prediction error for selecting
$\lambda$ and $\hat{\theta}_{\lambda}^{K}$, we need to have a good
estimate of the initial condition $x_{0}$ as it is necessary for
computing the predictions. We propose then to select $\lambda$ by
minimizing the Sum of Squared Errors 
\begin{equation}
SSE(\lambda)=\sum_{i=1}^{n}\left\Vert Y_{i}-CX_{\widehat{\theta}_{\lambda},\widehat{x_{0,\lambda}},0}(t_{i})\right\Vert _{2}^{2}.\label{eq:sse_cv}
\end{equation}
Moreover, this criterion gives a way to reduce the influence of the
nonparametric estimate $\hat{Y}$, as we use the original noisy data.
This is the selection procedure that we implemented in the experiments
part.


\section{Experiments}

We use two test beds for evaluating the practical efficiency of the
deterministic Kalman filter estimator $\hat{\theta}^{K}$; we compare
it with the NLS estimator $\hat{\theta}^{NLS}$ and the estimator
obtained by Generalized Smoothing $\hat{\theta}^{GS}$. The two models
are linear in the states, and they can be linear or nonlinear w.r.t
parameters. We use several sample size and several variance error
for comparing robustness and efficiency.

\subsection{Experimental design}

For a given sample size $n$ and noise level $\sigma$, we estimate
the Mean Square Error and the mean Absolute Relative Error (ARE) $\mathbb{E}_{\theta^{\text{*}}}\left[\frac{\left|\theta^{*}-\widehat{\theta}\right|}{\left|\theta^{*}\right|}\right]$
by Monte Carlo, based on $N_{MC}=100$ runs. For each run, we simulate
an ODE solution with a Runge-Kutta algorithm (ode45 in Matlab), and
a centered Gaussian noise (with variance $\sigma$) is added, in order
to obtain the $Y_{i}$'s. We compare the accuracy of the 3 parameters
$\hat{\theta}^{K}$ ,$\widehat{\theta}^{GS}$ and $\widehat{\theta}^{NLS}$,
but we are also interested in their mean prediction error defined
as 
\begin{equation}
E_{P}\left(\hat{X}\right)=\mathbb{E}_{(Y_{1},\dots,Y_{n})}\left[\mathrm{\mathbb{E}_{\theta^{\text{*}},\sigma}}\left[\left\Vert Y^{*}-\hat{X}\right\Vert \right]\right]\label{eq:prediction_error}
\end{equation}
where $Y^{*}$ is a new observation generated with the parameters
$(\theta^{*},x_{0}^{*},\sigma)$, and $\hat{X}$ is an estimator of
the trajectory, based on one of the three estimates $\hat{\theta}^{K}$
,$\widehat{\theta}^{GS}$ and $\widehat{\theta}^{NLS}$. For the three
estimators, the initial condition is estimated consistently: 
\begin{description}
\item [{NLS:}] $\widehat{x_{0}}$ is obtained simultaneously with the parameter
estimation (as an additional parameter), 
\item [{Kalman:}] $\widehat{x_{0}}$ and $\lambda_{n}$ are selected as
described in section 5.3, 
\item [{Generalized~~Smoothing:}] $\widehat{x_{0}}$ is the initial value
of the estimated curve corresponding to the estimated parameter $\widehat{\theta}^{GS}$,
with smoothing parameter $\lambda_{n}$ selected adaptively as described
in \cite{Ramsay2007}.
\end{description}
We insist on the fact that parameter estimation and prediction are
two different statistical tasks, that are evaluated by different criteria.
Parameter estimation is required when the parameter has an interest
by itself or when the model has an explicative purpose, whereas the
prediction error is dedicated to estimation of the state $X$, in
the most efficient way. Our primary interest is parameter estimation
but we also discuss prediction for the three methods; as we have seen
in section 5, parameter estimation and state estimation are tightly
related in particular for the selection of $\lambda$. We will consider
two possible estimators for the state: a parametric estimator $X_{\hat{\theta},\widehat{x}_{0}}$
and a smoothed (or corrected) estimator $X_{\hat{\theta},\widehat{x}_{0},\bar{u}}$.
For the NLS estimator, the parametric and corrected state estimator
are the same, whereas $X_{\hat{\theta}^{GS},\widehat{x}_{0}}$ and
$\hat{X}(\cdot,\hat{\theta}^{GS})$ are different, as $X_{\hat{\theta},\widehat{x}_{0}}$
differs from $X_{\hat{\theta},\widehat{x}_{0},\bar{u}}$. 

The two test beds are partially observed models with one missing state
variable. We compare the ability of the different methods to accurately
reconstruct the hidden state. Thus, we compute for each estimator
the $L^{2}-$distance between the true missing state and the obtained
reconstruction after parameter estimation:
\begin{equation}
\Delta\left(\hat{X}^{unobs}\right)=\mathrm{\mathbb{E}}\left[\left\Vert X_{\theta^{*},x_{0}^{*}}^{unobs}-\hat{X}^{unobs}\right\Vert _{L^{2}}\right]\label{eq:l2_discrepancy}
\end{equation}
The nonparametric estimate $\hat{Y}$ is a regression spline, with
a B-spline basis defined on a uniform knot sequence $\xi_{k},k=1,\dots,K$.
For each run and each state variables, the number of knots is selected
by minimizing the GCV criterion, \cite{ruppert2003semiparametric}.
For optimizing the criterion $S$, we use the Matlab function 'fminunc'
that implements a trust region algorithm for which gradient expression
is required. The computation of the gradient of $S$ w.r.t the parameter
$\theta$ is computationally involved and is based on the sensitivity
equations of the ODE model. The computational details are left in
appendix \ref{sec:Gradient-Computation}.

\subsection{Toy Examples: Partially Observed ODE in 3 D}

We consider the autonomous ODE 
\begin{equation}
\left\{ \begin{array}{l}
\dot{x_{1}}=-(k_{1}+k_{2})x_{1}\\
\dot{x_{2}}=k_{1}x_{1}\\
\dot{x_{3}}=k_{2}x_{1}
\end{array}\right.\label{eq:part_obs_model}
\end{equation}
where we observe only the variables $x_{2}$ and $x_{3}$. Using the
notation introduced in this paper, we have 
\[
A_{\theta}=\left(\begin{array}{ccc}
-(k_{1}+k_{2}) & 0 & 0\\
k_{1} & 0 & 0\\
k_{2} & 0 & 0
\end{array}\right)
\]
and 
\[
C=\left(\begin{array}{ccc}
0 & 1 & 0\\
0 & 0 & 1
\end{array}\right).
\]
With that model, we show that the conditions introduced in the statistical
analysis are workable on some simple models, in particular the conditions
for identifiability C2a and C2b that needs to be checked.  In the
case of autonomous system (i.e when $A_{\theta}$ and $r_{\theta}$
do not depend on time), a simple sufficient and necessary criteria
is the so-called Kalman criterion: 
\begin{prop}
\label{prop:autonomous_invertibility_criteria-1} In the case of an
autonomous model, the matrix $E_{\theta}(T)$ is nonsingular if and
only if the matrix 
\begin{equation}
K_{A,C}=\left(\begin{array}{l}
C\\
CA_{\theta}\\
\vdots\\
CA_{\theta}^{d-1}
\end{array}\right)\label{eq:KalmanMatrix}
\end{equation}
has a rank equals to $d$. 
\end{prop}
The matrix $K_{A,C}$ is usually called the Kalman matrix. In order
to define properly our criterion $S$, we need to check that condition
2b is also satisfied (joint identifiability of $\theta$ and $x_{0}$).
For this model, the analysis is relatively easy and we can use the
characterisation proposed by \cite{Pohjanpalo1978} based on the power
series expansion. As, the Kalman matrix (\ref{eq:KalmanMatrix}) is
\[
\left(\begin{array}{l}
C\\
CA_{\theta}\\
CA_{\theta}^{2}
\end{array}\right)=\left(\begin{array}{ccc}
0 & 1 & 0\\
0 & 0 & 1\\
k_{1} & 0 & 0\\
k_{2} & 0 & 0\\
-k_{1}(k_{1}+k_{2}) & 0 & 0\\
-k_{2}(k_{1}+k_{2}) & 0 & 0
\end{array}\right),
\]
the Kalman condition is fullfilled (i.e the matrix rank is $3$) if
$k_{1}\neq0$ or $k_{2}\neq0$. Hence, C2a holds for all relevant
cases ( $k_{1}=0$ or $k_{2}=0$ correspond to the case where $x_{2}$
and $x_{3}$ variations are disconnected from $x_{1}$ which makes
the model useless for explanation or prediction purposes). 

For condition C2b, we use the result shown by  Pohjanpalo et al. If
the model is $\dot{x}=f(t,x,\theta)$ and the observation function
is $h(t,\theta,x)$, condition C2b is satisfied if the nonlinear system
\begin{equation}
h^{(j)}(x_{0},\theta,x)=a_{j}(x_{0})\: j=0,1,\ldots\label{eq:PhojanpaloSystem}
\end{equation}
 has a unique solution $\theta$. Pohjanpalo et al. showed that for
linear autonomous system, this condition is sufficient and necessary.
In our case, the equation (\ref{eq:PhojanpaloSystem}) can be written
as 
\[
CA_{\theta}^{j}X_{0}=a_{j}\: j=0,1,\ldots
\]
Since the initial condition $X_{0}=\left(X_{0,1},X_{0,2},X_{0,3}\right)$
is unknown, we have to consider the extended parameter $\theta=(k_{1}k_{2},X_{0}^{T})$.
The equations for $j=0$ 
\[
CX_{0}=a_{j}\: j=0,1,\ldots
\]
 allow us to identify $X_{0,2}=a_{0,2}\,,X_{0,3}=a_{0,3}$. For $j=1$,
we have 
\[
\left\{ \begin{array}{lll}
k_{1}X_{0,1} & = & a_{1,2}\\
k_{2}X_{0,1} & = & a_{1,3}
\end{array}\right.
\]
and the solutions are $X_{0,1}=\frac{a_{1,3}}{k_{2}}$ and $k_{1}=\frac{a_{1,2}}{a_{1,3}}k_{2}$.
Finally, we have a unique solution for $k_{2}$, if we consider the
additional equation (\ref{eq:PhojanpaloSystem}) for $j=2$. In that
case, the system 
\[
\left\{ \begin{array}{lll}
-k_{1}(k_{1}+k_{2})X_{0,1} & = & a_{2,2}\\
-k_{2}(k_{1}+k_{2})X_{0,1} & = & a_{2,3}
\end{array}\right.
\]
has a unique solution $k_{2}=-\frac{a_{2,3}}{a_{1,2}+a_{1,3}}$.

\paragraph{Well-specified model (Toy Model 1)}

We test two sample sizes $n=200$ and $n=100$ (observations times
are uniformely sampled between $t=0$ and $t=100$) and two noise
levels $\sigma=3$ and $\sigma=6$. For the computation of the regression
splines $\hat{Y}$, we select manually the knots location instead
of using the GCV driven selection (to avoid overfitting). We have
placed four equispaced knots respectively at time $t=0,\,33,\,66$
and $100$. The true parameter is $\theta^{*}=\left(k_{1}^{*},k_{2}^{*}\right)=\left(0.0593,0.0296\right)$
and the initial condition $x_{0}^{*}$ equals $\left(0,\,0,\,100\right)$.
For the Kalman estimator, we select $\lambda_{n}$ by cross-validation
among the values $\lambda^{v}=\left\{ 10^{k}\right\} _{k\in\left[5\,16\right]}$. 

\begin{center}
{\small{}}
\begin{table}[h]
\centering{}{\small{}}%
\begin{tabular}{|c|c|c|c|c|c|}
\hline 
{\small{}$\left(n,\sigma\right)$} &  & \multicolumn{1}{c|}{{\small{}MSE ($10^{-6}$)}} & {\small{}ARE ($10^{-2}$)} & $E_{P}\left(X_{\widehat{\theta},\widehat{x_{0}}}\right)$ & $\Delta\left(X_{\widehat{\theta},\widehat{x_{0}}}\right)$\tabularnewline
\hline 
\hline 
\multirow{3}{*}{{\small{}$\left(200,3\right)$}} & {\small{}$\widehat{\theta}^{NLS}$} & 4.20 & 5.16 & 43.56 & 4.18\tabularnewline
\cline{2-6} 
 & {\small{}$\widehat{\theta}^{K}$} & 3.97 & 4.77 & 42.79 & 4.13\tabularnewline
\cline{2-6} 
 & $\widehat{\theta}^{GS}$ & 12.87 & 12.23 & 46.80 & 9.28\tabularnewline
\hline 
\multirow{3}{*}{{\small{}$\left(200,6\right)$}} & {\small{}$\widehat{\theta}^{NLS}$} & 17.09 & 9.92 & 87.76 & 8.43\tabularnewline
\cline{2-6} 
 & {\small{}$\widehat{\theta}^{K}$} & 16.49 & 9.43 & 85.45 & 8.28\tabularnewline
\cline{2-6} 
 & $\widehat{\theta}^{GS}$ & 77.87 & 23.28 & 93.69 & 17.77\tabularnewline
\hline 
\multirow{3}{*}{{\small{}$\left(100,3\right)$}} & {\small{}$\widehat{\theta}^{NLS}$} & 8.21 & 7.43 & 44.95 & 6.04\tabularnewline
\cline{2-6} 
 & {\small{}$\widehat{\theta}^{K}$} & 8.78 & 7.37 & 43.03 & 6.15\tabularnewline
\cline{2-6} 
 & $\widehat{\theta}^{GS}$ & 22.32 & 12.60 & 48.01 & 9.45\tabularnewline
\hline 
\multirow{3}{*}{{\small{}$\left(100,6\right)$}} & {\small{}$\widehat{\theta}^{NLS}$} & 36.89 & 15.27 & 90.76 & 12.24\tabularnewline
\cline{2-6} 
 & {\small{}$\widehat{\theta}^{K}$} & 34.98 & 14.91 & 86.19 & 12.36\tabularnewline
\cline{2-6} 
 & $\widehat{\theta}^{GS}$ & 86.74 & 24.39 & 94.91 & 18.63\tabularnewline
\hline 
\end{tabular}{\small{}\protect\caption{\label{tab:exp_res_part_obs_model_wellspe}Results for the Toy Model
1 ; partially observed.}
}
\end{table}

\par\end{center}{\small \par}

The results are presented in table \ref{tab:exp_res_part_obs_model_wellspe}.
The GS estimator is outperformed by the Kalman and NLS estimators,
moreover our approach improves the parameter estimation accuracy in
terms of MSE and ARE in almost every cases comparing to the NLS and
also minimizes prediction error. Regarding the missing state reconstruction
both methods gives similar results.

\paragraph{Misspecified model (Toy Model 2)}

In our simulation, we give also some insight in the case of misspecified
models. Indeed, our perturbed ODE framework permits to consider naturally
the problem of model misspecification, when the true model is $\dot{x}(t)=A_{\theta}(t)x(t)+r_{\theta}(t)+v(t)$,
with $v\in L^{2}(\left[0,\, T\right],\mathbb{R}^{d})$ an unknown
function. We do not provide any theoretical analysis for this kind
of model misspecification. The Kalman estimator gives more accurate
estimation than the NLS estimator in that case, as we consider pertubations
of the initial model. Moreover, the optimal control $\overline{u}$
obtained along the parameter estimation can be used as a correction
term to add to the initial model to counter-balance misspecification.
This implies potentially a better prediction power. The true model
is nearly the same model as above 

\begin{equation}
\dot{X}=A_{\theta}X+v(t)\label{eq:true_model}
\end{equation}
with $\theta^{*}=\left(k_{1}^{*},k_{2}^{*}\right)=\left(0.0593,0.0296\right)$
and $x_{0}^{*}=\left(0,\,0,\,100\right)$, but we add a pertubation
$v:\left[0,\, T\right]\longmapsto\mathbb{R}^{3}$ with entries equal
to $0.4\times\sin(\frac{t}{5})$. Nevertheless for parameter estimation,
we still use the unperturbed model $\dot{X}=A_{\theta}X$. 

In the case of the Kalman estimator $\hat{\theta}^{K}$, the optimal
control $\overline{u}$ can be used for correcting the model and for
defining a new model 
\begin{equation}
\dot{X}=A_{\widehat{\theta}}X+\overline{u}.\label{eq:corrected_model}
\end{equation}
We are then interested in evaluating the prediction error of $X_{\widehat{\theta},\widehat{x_{0}},\overline{u}}$,
defined as $E_{P}\left(X_{\widehat{\theta},\widehat{x_{0}},\overline{u}}\right)$.
 We also estimate the error between the true first state value and
the obtained reconstruction with the corrected model. As shown in
the introduction, Generalized Smoothing can also evaluate a correction
term for $\widehat{\theta}^{GS}$, defined as $\overline{u}(t)=\dot{\widehat{X}}(t,\widehat{\theta}^{GS})-A_{\widehat{\theta}}\widehat{X}(t,\widehat{\theta}^{GS})$
(where $\widehat{X}(t,\widehat{\theta}^{GS})$ is the spline corresponding
to the estimated parameter $\widehat{\theta}^{GS}$ with adaptive
$\widehat{\lambda}$). In the case of NLS, we cannot compute a correction
$\bar{u}$, as the estimated trajectories are exactly solution of
the ODE for $\widehat{\theta}^{NLS}$. In the case of Generalized
Smoothing, we have $\Delta\left(X_{1}^{*};X_{\widehat{\theta},\widehat{x_{0}}}\right)\approx\Delta\left(X_{1}^{*};X_{\widehat{\theta},\widehat{x_{0}},\bar{u}}\right)$
because the hidden parts are (almost) exactly trajectories of the
ODE with parameter $\hat{\theta}^{GS}$. The estimates that change
is the Kalman-based one. 

\begin{center}
{\small{}}
\begin{table}[h]
\centering{}{\small{}}%
\begin{tabular}{|c|c|c|c|c|c|c|c|}
\hline 
{\small{}$\left(n,\sigma\right)$} &  & \multicolumn{1}{c|}{{\small{}MSE ($10^{-5}$)}} & {\small{}ARE ($10^{-2}$)} & $E_{P}\left(X_{\widehat{\theta},\widehat{x_{0}}}\right)$ & $E_{P}\left(X_{\widehat{\theta},\widehat{x_{0}},\bar{u}}\right)$ & $\Delta\left(X_{\widehat{\theta},\widehat{x_{0}}}^{unobs}\right)$ & $\Delta\left(X_{\widehat{\theta},\widehat{x_{0}},\bar{u}}^{unobs}\right)$\tabularnewline
\hline 
\hline 
\multirow{3}{*}{{\small{}$\left(200,3\right)$}} & {\small{}$\widehat{\theta}^{NLS}$} & {\small{}4.14} & {\small{}19.13} & 52.24 & 52.24 & 19.08 & 19.08\tabularnewline
\cline{2-8} 
 & {\small{}$\widehat{\theta}^{K}$} & {\small{}3.66} & {\small{}17.46} & 47.78 & 47.75 & 18.82 & 18.87\tabularnewline
\cline{2-8} 
 & $\widehat{\theta}^{GS}$ & 7.56 & 27.34 & 55.13 & 50.99 & 22.27 & 22.26\tabularnewline
\hline 
\multirow{3}{*}{{\small{}$\left(200,6\right)$}} & {\small{}$\widehat{\theta}^{NLS}$} & 4.99 & 18.65 & 92.95 & 92.95 & 19.90 & 19.90\tabularnewline
\cline{2-8} 
 & {\small{}$\widehat{\theta}^{K}$} & 4.68 & 18.14 & 88.30 & 88.02 & 20.66 & 19.79\tabularnewline
\cline{2-8} 
 & $\widehat{\theta}^{GS}$ & 13.21 & 29.66 & 97.25 & 94.64 & 26.09 & 26.08\tabularnewline
\hline 
\multirow{3}{*}{{\small{}$\left(100,3\right)$}} & {\small{}$\widehat{\theta}^{NLS}$} & 4.88 & 19.56 & 52.66 & 52.66 & 19.48 & 19.48\tabularnewline
\cline{2-8} 
 & {\small{}$\widehat{\theta}^{K}$} & 4.56 & 18.53 & 48.07 & 47.86 & 19.76 & 19.23\tabularnewline
\cline{2-8} 
 & $\widehat{\theta}^{GS}$ & 10.04 & 29.32 & 56.33 & 55.06 & 23.71 & 23.66\tabularnewline
\hline 
\multirow{3}{*}{{\small{}$\left(100,6\right)$}} & {\small{}$\widehat{\theta}^{NLS}$} & 7.96 & 23.32 & 96.63 & 96.63 & 21.69 & 21.69\tabularnewline
\cline{2-8} 
 & {\small{}$\widehat{\theta}^{K}$} & 7.59 & 22.36 & 89.19 & 88.65 & 23.88 & 21.56\tabularnewline
\cline{2-8} 
 & $\widehat{\theta}^{GS}$ & 15.63 & 32.77 & 101.41 & 98.81 & 26.14 & 26.15\tabularnewline
\hline 
\end{tabular}{\small{}\protect\caption{\label{tab:exp_res_part_obs_model_misspe}Results for Toy Model 2,
partially observed model; misspecified case}
}
\end{table}

\par\end{center}{\small \par}

The GS parameter estimator is outperformed by the Kalman and the NLS
estimator. Our approach improves the estimation accuracy for $\theta$
(lower MSE and ARE in every cases)on the NLS estimator. This difference
is bigger than in the well specified case (Toy Model 1), as we are
more robust to the presence of a perturbation than the NLS. The Kalman
estimator gives also better prediction error in every cases and the
correction $\overline{u}$ slightly improves the prediction errors.
Nevertheless, the NLS estimator provides the smallest $\Delta\left(X_{1}^{*};X_{\widehat{\theta},\widehat{x_{0}}}\right)$
among all estimation methods in every cases but the first one. Nonetheless
using $\overline{u}$ minimizes in most of case the error for $X_{1}$
estimation for our approach and allows us to obtain slightly better
result than the NLS estimator. 

The correction term $\bar{u}$ is related (correlated) to the perturbation
$t\mapsto v(t)$ as we can in figure \ref{fig:mean_residual_control},
where we plot the mean of each component of $\overline{u}$, when
$\left(n,\sigma\right)=\left(200,\,3\right)$. Even though the scale
is not the same (we need to rescale by $10^{-5}$ for easing comparisons),
the correction $\bar{u}$ exhibits some important features of the
true one, such as oscillations with a period close to the period of
$v$. The analysis of $\bar{u}$ is beyond the scope of that paper,
but the presence of strong patterns in $u$ can be used to detect
misspecification, in the same way that the analysis of residuals permits
to detect lack of fit in regression models.

\begin{center}
\begin{figure}[h]
\centering{}\includegraphics[scale=0.5]{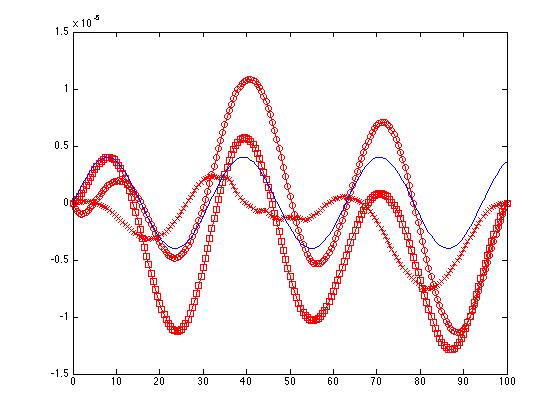}\protect\caption{\label{fig:mean_residual_control}Toy Model 2, $\left(n,\sigma\right)=\left(200,\,3\right)$:
Mean correction $\bar{u}$ (red curve); rescaled true perturbation
$10^{-5}\times v$ (blue curve)}
\end{figure}

\par\end{center}

\subsection{Real case example: Methanation reaction}

We consider an ODE model introduced in \cite{Happel1980} for describing
the dynamics of carbon monoxide and hydrogen methanation over a supported
nickel catalyst by transient isotopic tracer in a gradientless circulating
reactor. This ``Methanation reaction'' model is a linear autonomous
equation in $\mathbb{R}^{4}$, with a forcing term. A important difference
w.r.t the previous is the nonlinearity in parameters as we have 
\[
A_{\theta}=\left(\begin{array}{cccc}
-\frac{V+V'+F_{0}^{C0}/W}{\beta C^{C0}/W+C^{COl}} & 0 & 0 & 0\\
\frac{V+V'}{\beta C^{H_{2}0}/W} & -\frac{V+V'+v_{5}}{\beta C^{H_{2}0}/W} & 0 & \frac{v_{5}}{\beta C^{H_{2}0}/W}\\
\frac{V'}{\beta C^{CO_{2}}/W} & 0 & -\frac{V'+v_{6}}{\beta C^{CO_{2}}/W} & \frac{v_{6}}{\beta C^{CO_{2}}/W}\\
0 & \frac{v_{5}}{C^{O_{s}}} & \frac{v_{6}}{C^{O_{s}}} & -\frac{v_{5}+v_{6}}{C^{O_{s}}}
\end{array}\right)
\]
 and $r_{\theta}=\left(\frac{F_{i}^{C0}z_{i}^{CO}}{\beta C^{C0}/W+C^{COl}},\,0,\,0,\,0\right)^{\top}$.
{\small{}The state $X$ is defined as $X^{\top}=\left(X^{CO},X^{H_{2}O},X^{CO_{2}},X^{O_{s}}\right)$,
and represents the quantity o}f the chemical species involved in the
reaction.  A constant inlet $CO$ flow rate with constant and known
fraction of isotope $^{18}O$ is introduced within the reactor; the
fraction of $^{18}O$ present in oxygen atoms for each component is
measured at different timeframe using a mass spectrometer. In the
model, $X^{j}(t)$ represents the measured fraction of $^{18}O$ present
in oxygen atoms of the chemical species $j$ at time $t$. The total
amount of oxygen $X^{O_{s}}$ cannot be measured. Some of the parameter
are already known:
\begin{itemize}
\item $F_{i}^{CO}/F_{0}^{CO}$ : inlet/outlet flow rates of CO ($0.59/0.45$)
\item $z_{i}^{CO}$ : the constant fraction of $^{18}O$ present in oxygen
atoms of the CO inlet flow rate ($0.132$)
\item $V/V'$: rates of production ($0.124/0.01$)
\item $C^{j}$: concentrations of gas phases in the reaction system ($j=CO,\, H_{2}0,\, CO_{2}$)
\item $W$ total weight of catalyst within system ($0.744$)
\item $\beta$ volume of dead space ($206.1$)
\end{itemize}
Our aim is to estimate the parameter $\theta=\left(C^{COl},C^{Os},v_{5},v_{6}\right)$. 

For simulating the datasets, we use two sample sizes $n=100$ and
$n=50$ (observations are uniformely sampled the time interval $\left[0,40\right]$),
with 2 noise levels $\sigma=0.002$ and $\sigma=0.004$. The true
parameter value is the estimate provided in \cite{Happel1980}, i.e
$\theta^{*}=\left(0.1,\,11.1,\,0.35,\,0.008\right)$ and with initial
condition equals to $x_{0}^{*}=\left(0,\,0,\,0,\,0\right)$. For the
computation of the Kalman estimator, we select $\lambda$ among $1,\,5,\,20,\,50,\,100,\,200,\,$$300,\,400,\,500,\,600,\,700,\,800,\,900,\,1000$.
Finally, the nonparametric estimate $\hat{Y}$ is a regression splines,
with knots selected manually (instead of GCV selection, because of
overfitting): we use three equispaced knots at times $t=0,\,20,\,40$.

\begin{center}
{\small{}}
\begin{table}[h]
\centering{}{\small{}}%
\begin{tabular}{|c|c|c|c|c|c|c|c|}
\hline 
{\small{}$\left(n,\sigma\right)$} &  & \multicolumn{1}{c|}{{\small{}MSE }} & {\small{}ARE } & $E_{P}\left(X_{\widehat{\theta},\widehat{x_{0}}}\right)$ & $E_{P}\left(X_{\widehat{\theta},\widehat{x_{0}},\bar{u}}\right)$ & $\Delta\left(X_{\widehat{\theta},\widehat{x_{0}}}^{unobs}\right)$ & $\Delta\left(X_{\widehat{\theta},\widehat{x_{0}},\bar{u}}^{unobs}\right)$\tabularnewline
\hline 
\hline 
\multirow{3}{*}{{\small{}$\left(100,0.002\right)$}} & {\small{}$\widehat{\theta}^{NLS}$} & 17.28 & 1.09 & 19.43 & 19.43 & 6.15 & 6.15\tabularnewline
\cline{2-8} 
 & {\small{}$\widehat{\theta}^{K}$} & 3.60 & 1.06 & 8.22 & 1.11 & 2.16 & 0.70\tabularnewline
\cline{2-8} 
 & $\widehat{\theta}^{GS}$ & 21.54 & 1.14 & 19.45 & 19.45 & 6.41 & 6.41\tabularnewline
\hline 
\multirow{3}{*}{{\small{}$\left(50,0.002\right)$}} & {\small{}$\widehat{\theta}^{NLS}$} & 57.23 & 2.19 & 31.49 & 31.49 & 12.62 & 12.62\tabularnewline
\cline{2-8} 
 & {\small{}$\widehat{\theta}^{K}$} & 21.38 & 2.05 & 9.26 & 3.49 & 2.58 & 1.38\tabularnewline
\cline{2-8} 
 & $\widehat{\theta}^{GS}$ & 58.05 & 2.11 & 44.40 & 44.39 & 12.40 & 12.40\tabularnewline
\hline 
\multirow{3}{*}{{\small{}$\left(100,0.004\right)$}} & {\small{}$\widehat{\theta}^{NLS}$} & 50.98 & 1.55 & 41.60 & 41.60 & 12.00 & 12.00\tabularnewline
\cline{2-8} 
 & {\small{}$\widehat{\theta}^{K}$} & 26.76 & 1.44 & 7.54 & 2.11 & 2.50 & 1.35\tabularnewline
\cline{2-8} 
 & $\widehat{\theta}^{GS}$ & 55.61 & 1.59 & 33.66 & 33.66 & 15.96 & 15.96\tabularnewline
\hline 
\multirow{3}{*}{{\small{}$\left(50,0.004\right)$}} & {\small{}$\widehat{\theta}^{NLS}$} & 80.03 & 2.25 & 43.13 & 43.13 & 14.75 & 14.75\tabularnewline
\cline{2-8} 
 & {\small{}$\widehat{\theta}^{K}$} & 35.87 & 2.16 & 28.69 & 3.06 & 7.79 & 1.57\tabularnewline
\cline{2-8} 
 & $\widehat{\theta}^{GS}$ & 94.30 & 2.29 & 44.59 & 44.59 & 17.86 & 17.86\tabularnewline
\hline 
\end{tabular}{\small{}\protect\caption{\label{tab:exp_res_methanation}Methanation Model}
}
\end{table}

\par\end{center}{\small \par}

The results are presented in table \ref{tab:exp_res_methanation},
that gathers the statistics about the parameter estimation accuracy,
and the prediction of the complete state, and in particular the estimation
of the hidden variable $X^{O_{s}}$.  The Kalman estimator gives
more accurate parameter estimates than Nonlinear Least Squares or
Generalized Smoothing. The dramatic difference for the MSE comes from
the estimation of $C^{O_{s}}$ that is of greater magnitude than the
other parameters, thus ARE seems more relevant for comparisons. However,
the MSE enlighten the difficulty for NLS and GS estimator to correctly
estimate $C^{Os}$; moreover, a great number of outliers for $C^{Os}$
estimates have been removed for the NLS estimation before computing
ARE and MSE. Additionaly, state estimation improves dramatically,
as the prediction error $E_{P}$ and missing state reconstruction
$\Delta$ of the Kalman estimator outperforms the two others. This
improvement is even more significative when the correction $\bar{u}$
is used. 

The difference can be partly explained by the nonlinearity in parameters
that makes their estimation more difficult. We can have estimates
that are far from the true parameter value, but in the case of the
Kalman estimator, the important errors for the parameter are balanced
by a more important correction term $\bar{u}$ that ameliorates significantly
state estimation and prediction. 

\section{Discussions}

We have considered the statistical problem of parameter and state estimation of a linear Ordinary Differential Equations as an Optimal Control problem. By doing this, we follow the lines drawn in \cite{Ramsay2007} or in the two-step approaches, that consist in defining a statistical criterion more adapted to ordinary differential equation than the likelihood. A new theory was needed in order to assess the statistical efficiency of this new estimator, that heavily relies on the Linear-Quadratic Theory. Indeed, the linear structure of the model gives a closed-form for the criterion $S$ which permits to establish the needed regularity properties for statistical analysis. An important question is to determine the conditions under which we can apply the same methodology for nonlinear ODEs. It is probably more involved but the characterization used here is directly generalized by the Pontryagin's Maximum Principle, that gives also a tractable way to solve the optimal control problem.      

An important feature of our approach is that we can cope with model misspecification, and the estimation process gives a way to evaluate the lack-of-fit thanks to the analysis of the control $\bar{u}$. Thanks to that, we are able to estimate properly the parameters, but also to do prediction and state estimation. Our experiments show that we can have better performance than the classical NLS and Generalized Smoothing and that it is beneficial to account for possible perturbation. A good choice for the trade-off hyperparameter $\lambda$ is then necessary, and our selection methodology is satisfying in practice but needs more insight to explain its influence for the selection of good predictors, in particular for hidden states. The penalty term $\Vert u \Vert_{L^2}$ is an energy related to the degrees of freedom of the predictor  $X_{\theta,x_0,u}$, but it is not related to the usual  criterion of model complexity for smoothing. 

In our analysis, we assume that the observability $O_\theta(T)$ is nonsingular, which avoids the use of the quadratic form $x_0\top Q x_0$ in the criterion. If this regularization term is used, the mechanics of the proof would be the same, but with  $Q=Q_n$ that should tend to $0$ , as $n$ tends to infinity. Nevertheless, it can have consequences on the asymptotics of the estimators, as it corresponds to cases where the loss of information is too big and needs additional information. Quite interestingly, our criterion about identifiability remains tractable, and can be relatively easy to check in practice.

\newpage{}

\appendix

\part*{Appendix: State and Parameter Inference for Partially Observed ODE}

\section{Derivation of deterministic Kalman filter estimator using Linear-Quadratic
Theory \label{sec: DerivationKalman}}

In this section we describe more precisely how the deterministic Kalman
Filter is constructed (see \cite{Sontag1998} for an introduction),
it involves two steps: 
\begin{enumerate}
\item For a given initial condition $x_{0}$ we determine the minimum
cost expression thanks to theorem \ref{thm:LQ_existence_unicity}
(subsection \ref{sub:x0-fixed}). 
\item  inimal cost is a quadratic form w.r.t final
condition and hence it exists an unique final condition (and hence
a unique initial condition by ODE solution uniqueness) minimizing
this minimal cost (subsection \ref{sub:Optimal_xo}). 
\end{enumerate}

\subsection{\label{sub:x0-fixed}$x_{0}$ fixed, minimal cost expression}
To derive a closed form for the minimal cost
for a given $x_{0}$. For that we define the reverse time functions:
\begin{equation}
\begin{array}{l}
\widetilde{X}_{\theta,x_{0},u}(t)=X_{\theta,x_{0},u}(T\text{\textminus}t),\,\widetilde{A_{\theta}}(t)=-A_{\theta}(T\text{\textminus}t)\\
\widetilde{r_{\theta}}(t)=-r_{\theta}(T-t),\,\widetilde{B}(t)=-B(T\text{\textminus}t),\,\widetilde{Y}(t)=\widehat{Y}(T\text{\textminus}t)
\end{array}\label{eq:gen_writing_function}
\end{equation}

And by denoting 
\begin{equation}
\widetilde{W_{1}}=\left(\begin{array}{cc}
C^{T}C & -C^{T}\widetilde{Y}\\
-\widetilde{Y}^{T}C & \widetilde{Y}^{T}\widetilde{Y}
\end{array}\right),\, Q_{1}=\left(\begin{array}{cc}
Q & 0\\
0 & 0
\end{array}\right),Z_{\theta,x_{0},u}=\left(\begin{array}{c}
\widetilde{X}_{\theta,x_{0},u}\\
1
\end{array}\right)\label{eq:gen_writing_weight_matrix}
\end{equation}
we can rewrite our cost under the form: 
\begin{equation}
\begin{array}{lll}
C\left(\hat{Y};x_{0},u,\theta,\lambda\right) & = & \widetilde{C}(\hat{Y};Z_{\theta,x_{0},u}(T),u,\theta,\lambda)\\
 & := & Z_{\theta,x_{0},u}(T)^{T}Q_{1}Z_{\theta,x_{0},u}(T)+\int_{0}^{T}\left\Vert Z_{\theta,x_{0},u}(t)\right\Vert _{\widetilde{W}_{1}}^{2}dt+\lambda\int_{0}^{T}\left\Vert u(t)\right\Vert _{2}^{2}dt
\end{array}\label{eq:generalized_cost_function}
\end{equation}
The issue here is to minimize (\ref{eq:generalized_cost_function})
in a non-finite dimensional space but thanks to results coming from
Optimal control and Riccati theory we know that for a given $\theta$
and a given $Z(T)$ it exists a unique control $\bar{u}$ such that
\[
\widetilde{C}(\hat{Y};Z(T),\bar{u},\theta,\lambda)=\min_{u\in L^{2}}\widetilde{C}(\hat{Y};Z(T),u,\theta,\lambda)
\]
It is the main point of the following theorem for a given $\theta$
and $Z(T)$ it ensures the existence, the uniqueness of this control
$\bar{u}$ and gives a closed form for both $\bar{u}$ and $\widetilde{C}(\hat{Y};Z(T),\bar{u},\theta,\lambda)$. 
\begin{thm}
\label{thm:LQ_existence_unicity}Let $A\in L^{2}(\left[0,\, T\right],\mathbb{R}^{d\times d})$
and $B\in L^{2}(\left[0,\, T\right],\mathbb{R}^{d\times d})$ We consider
$z_{u}$ the solution of the following ODE: 
\[
\dot{z_{u}}(t)=A(t)z_{u}(t)+B(t)u(t),\: z(t_{0})=z_{0}
\]

and the cost: 
\[
C(t_{0},u,U)=z_{u}(T)^{T}Qz_{u}(T)+\int_{t_{0}}^{T}z_{u}(t)^{T}W(t)z_{u}(t)+u(t)^{T}U(t)u(t)dt
\]
with $Q$ positive,$W\in L^{\infty}(\left[0,\, T\right],\mathbb{R}^{d\times d})$
positive matrix for all $t\in\left[0,\, T\right]$ and $U(t)$ definite
positive matrix for all $t\in\left[0,\, T\right]$ respecting the
coercivity condition: 
\[
\exists\alpha>0\, s.t\,\forall u\in L^{2}(\left[0,\, T\right],\mathbb{R}^{d})\,:\,\int_{0}^{T}u(t)^{T}U(t)u(t)dt\geq\alpha\int_{0}^{T}\left\Vert u(t)\right\Vert _{2}^{2}dt
\]

For a given $t_{0}$ we want to minimize the cost $C(t_{0},u,U)$
on $L^{2}(\left[0,\, T\right],\mathbb{R}^{d})$.

We know it exists an unique control $\bar{u}$, called optimal control,
associated to the trajectory $z_{\bar{u}}$, called optimal trajectory,
minimizing this cost. Moreover $\bar{u}$ is under the closed-feedback
loop form $\overline{u}(t)=U^{-1}(t)E(t)B(t)z_{\overline{u}}(t)$
where $E$ is the matricial solution of the ODE:\textup{ 
\[
\begin{array}{l}
\dot{E}(t)=W(t)-A(t)^{t}E(t)-E(t)A(t)-E(t)B(t)U(t)^{-1}B(t)^{T}E(t)\\
E(T)=-Q
\end{array}
\]
}this ODE its called Ricatti equation associated to LQ problem composed
of the cost $C(t_{0},u,U)$ and the ODE $\dot{z_{u}}(t)=A(t)z_{u}(t)+B(t)u(t),\: z(t_{0})=z_{0}$.
Moreover $E(t)$ is symetric and the minimal cost is equal to: $C(t_{0},\overline{u},U)=-z_{0}^{T}E(t_{0})z_{0}$. 
\end{thm}
By identifying in the last theorem $A$ with $\left(\begin{array}{cc}
\widetilde{A_{\theta}}(t) & \widetilde{r_{\theta}}(t)\\
0 & 0
\end{array}\right)$, $Q$ with $Q_{1}$, $W$ with $W_{1}$ and $U$ with $\lambda I_{d}$
we obtain the corresponding minimal cost reached for the optimal cost
$\overline{u}$ for a given initial condition $x_{0}$
\begin{equation}
\widetilde{C}(\hat{Y};Z_{\theta,x_{0},\overline{u}}(T),\overline{u},\theta,\lambda)=-\widetilde{X}_{\theta,x_{0},\overline{u}}(0)^{T}\widetilde{E}_{\theta}(0)\widetilde{X}_{\theta,x_{0},\overline{u}}(0)\text{\textminus}2\widetilde{X}_{\theta,x_{0},\overline{u}}(0)^{T}\widetilde{h}_{\theta}(0)\text{\textminus}\widetilde{\alpha_{\theta}}(0).\label{eq:minimal_cost_x0_fixed}
\end{equation}
with the associated ODE:
\[
\begin{array}{l}
\dot{\widetilde{E_{\theta}}}(t)=C^{T}C-\widetilde{A}_{\theta}\widetilde{E_{\theta}}-\widetilde{E_{\theta}}\widetilde{A_{\theta}}-\frac{1}{\lambda}\widetilde{E_{\theta}}^{2},\:\widetilde{E_{\theta}}(T)=-Q\\
\dot{\widetilde{h_{\theta}}}(t)=-C^{T}\widetilde{Y}-\widetilde{A_{\theta}}^{T}\widetilde{h_{\theta}}-\widetilde{E_{\theta}}\widetilde{r_{\theta}}-\frac{1}{\lambda}\widetilde{E_{\theta}}\widetilde{h_{\theta}},\:\widetilde{h}_{\theta}(T)=0\\
\dot{\widetilde{\alpha_{\theta}}}=\widetilde{Y}^{T}\widetilde{Y}-2\tilde{r_{\theta}}^{T}\tilde{h_{\theta}}-\frac{1}{\lambda}\widetilde{h}_{\theta}^{T}\widetilde{h_{\theta}},\:\widetilde{\alpha_{\theta}}(T)=0
\end{array}
\]
To be able to apply Theorem \ref{thm:LQ_existence_unicity} we need
$W_{1}$ to belong to $L^{\infty}(\left[0,\, T\right],\mathbb{R}^{d\times d})$,
that is why we require $\widehat{Y}\in L^{\infty}(\left[0,\, T\right],\mathbb{R}^{d^{'}})$.

\subsection{\label{sub:Optimal_xo}Optimal $x_{0}$ selection}

For a given $x_{0}$ we have obtained the minimal cost expression
w.r.t control. How can we choose $x_{0}$ in order to minimize this
minimal cost?

We recall that $\widetilde{X}_{\theta,x_{0},\overline{u}}(0)=X_{\theta,x_{0},\overline{u}}(T)$
so $\widetilde{C}(\hat{Y};Z_{\theta,x_{0},\overline{u}}(T),\overline{u},\theta,\lambda)$
defined by (\ref{eq:minimal_cost_x0_fixed}) is a quadratic form w.r.t
the final condition ($\widetilde{\alpha_{\theta}}(0)$ do not depend
on $\widetilde{X}_{\theta,x_{0},\overline{u}}(0)$). Since it makes
no difference to minimize $\widetilde{C}(\hat{Y};Z_{\theta,x_{0},\overline{u}}(T),\overline{u},\theta,\lambda)$
w.r.t the final condition instead of $x_{0}$ because of unicity of
ODE solution we look for the final condition minimizing (\ref{eq:minimal_cost_x0_fixed}).
Hence if $\widetilde{E_{\theta}}(0)$ is invertible the minimum is
reached for 
\begin{equation}
-\widetilde{E_{\theta}}(0)^{-1}\widetilde{h}_{\theta}(0)\label{eq:optimal_CF}
\end{equation}
we denote $\widehat{x_{0}}$ the unique initial condition such that
$X_{\theta,\widehat{x_{0}},\overline{u}}(T)=-\widetilde{E_{\theta}}(0)^{-1}\widetilde{h}_{\theta}(0)$.

In that case the minimal cost is equal to:
\[
\widetilde{C}(\hat{Y};Z_{\theta,\widehat{x_{0}},\overline{u}}(T),\overline{u},\theta,\lambda)=\widetilde{h_{\theta}}(0)^{T}\widetilde{E_{\theta}}(0)^{-1}\widetilde{h}_{\theta}(0)\text{\textminus}\widetilde{\alpha}_{\theta}(0)
\]
and for a given parameter $\theta$ we have: 
\[
\begin{array}{lll}
S\left(\hat{Y};\theta,\lambda\right) & = & \min_{x_{0}\in\mathbb{R}^{d}}\left(\min_{u\in L^{2}}C(\hat{Y};x_{0},u,\theta,\lambda)\right)\\
 & = & \widetilde{h_{\theta}}(0)^{T}\widetilde{E_{\theta}}(0)^{-1}\widetilde{h}_{\theta}(0)\text{\textminus}\widetilde{\alpha}_{\theta}(0)\\
 & = & \widetilde{h_{\theta}}(0)^{T}\widetilde{E_{\theta}}(0)^{-1}\widetilde{h}_{\theta}(0)+\int_{0}^{T}\left(\widetilde{Y}(t)^{T}\widetilde{Y}(t)-2\tilde{r_{\theta}}(t)^{T}\widetilde{h_{\theta}}(t)-\frac{1}{\lambda}\widetilde{h_{\theta}}(t)^{T}\widetilde{h_{\theta}}(t)\right)dt
\end{array}
\]

\subsection{Minimal cost expression}

By posing $E_{\theta}(t)=-\widetilde{E_{\theta}}(T-t),\:\widehat{h_{\theta}}(t)=-\widetilde{h_{\theta}}(T-t)$
we define our estimator as: 
\[
\widehat{\theta}=\arg\min_{\theta\in\Theta}S\left(\hat{Y};\theta,\lambda\right)
\]

with the functional criteria: 
\begin{equation}
\begin{array}{lll}
S\left(Y;\theta,\lambda\right) & = & -h_{\theta}(T,Y)^{T}E_{\theta}(T)^{-1}h_{\theta}(T,Y)\\
 & + & \int_{0}^{T}\left(Y(t)^{T}Y(t)-2r_{\theta}(t)^{T}h_{\theta}(t,Y)-\frac{1}{\lambda}h_{\theta}(t,Y)^{T}h_{\theta}(t,Y)\right)dt
\end{array}\label{eq:criteria_function-1}
\end{equation}
the associated ODE: 
\begin{equation}
\begin{array}{l}
\dot{E_{\theta}}(t)=C^{T}C-A_{\theta}^{T}E_{\theta}-E_{\theta}A_{\theta}-\frac{1}{\lambda}E_{\theta}^{2},\\
\dot{h_{\theta}}(t,Y)=-\alpha_{\theta}(t)h_{\theta}(t,Y)-\beta_{\theta}(t,Y)\\
\left(E_{\theta}(0),h_{\theta}(0,Y)\right)=\left(Q,0\right)
\end{array}\label{eq:Riccati_ode-1}
\end{equation}
and the functions $\alpha$ and $\beta$ defined by: 
\[
\left\{ \begin{array}{l}
\alpha_{\theta}(t)=\left(A_{\theta}(t)^{T}+\frac{E_{\theta}(t)}{\lambda}\right)\\
\beta_{\theta}(t,Y)=C^{T}Y+E_{\theta}r_{\theta}
\end{array}\right.
\]
Hence we have obtained the expression for the optimal control, the
minimal cost and the final state value presented in Theorem \ref{thm:kalman_existence_unicity}.

\section{State Estimation: Controls of the variations of the adjoint variables}
\begin{lem}
\label{prop:E_discrepancy_wrt_lambda}We have $\left\Vert E_{\theta}(t)-E_{\theta'}(t)\right\Vert _{2}\leq K_{1}\overline{E_{\lambda}}e^{L_{1}\frac{\overline{E_{\lambda}}}{\lambda}}\left\Vert \theta-\theta^{'}\right\Vert $
by denoting $\overline{E_{\lambda}}=\sup_{t,\theta\in\left[0,\, T\right]\times\Theta}\left\Vert E_{\theta}(t)\right\Vert _{2}$
and $\overline{E_{\lambda}}\leq K_{2}e^{\frac{L_{1}}{\lambda}}$\end{lem}
\begin{proof}
Thanks to condition 3 $\forall\theta\in\Theta\: t\longmapsto A_{\theta}(t)$
is continuous on $\left[0\,,\, T\right]$ and $\forall\theta\in\Theta\: t\longmapsto E_{\theta}(t)$
is defined on $\left[0\,,\, T\right]$ and obviously continuous on
the same interval as an ODE solution.

$\forall(\theta,\theta')\in\Theta^{2}$ we have:

\[
\begin{array}{lll}
\dot{E_{\theta}}(t)-\dot{E_{\theta'}}(t) & = & A_{\theta'}(t)^{T}E_{\theta'}(t)-A_{\theta}(t)^{T}E_{\theta}(t)\\
 & + & E_{\theta'}(t)A_{\theta'}(t)-E_{\theta}(t)A_{\theta}(t)\\
 & + & \frac{1}{\lambda}\left(E_{\theta'}^{2}(t)-E_{\theta}^{2}(t)\right)
\end{array}
\]
and by integrating between $0$ and $t$, taking the norm gives us:
\[
\begin{array}{lll}
\left\Vert E_{\theta}(t)-E_{\theta'}(t)\right\Vert _{2} & \leq & \int_{0}^{t}\left\Vert A_{\theta'}(s)^{T}E_{\theta'}(s)-A_{\theta}(s)^{T}E_{\theta}(s)\right\Vert _{2}ds\\
 & + & \int_{0}^{t}\left\Vert E_{\theta'}(s)A_{\theta'}(s)-E_{\theta}(s)A_{\theta}(s)\right\Vert _{2}ds\\
 & + & \frac{1}{\lambda}\int_{0}^{t}\left\Vert E_{\theta'}^{2}(s)-E_{\theta}^{2}(s)\right\Vert _{2}ds
\end{array}
\]
and:
\[
\begin{array}{lll}
\int_{0}^{t}\left\Vert E_{\theta}^{2}(s)-E_{\theta'}^{2}(s)\right\Vert _{2}ds & \leq & \sqrt{d}\int_{0}^{t}\left\Vert E_{\theta}(s)\right\Vert _{2}\left\Vert E_{\theta}(s)-E_{\theta'}(s)\right\Vert _{2}ds\\
 & + & \sqrt{d}\int_{0}^{t}\left\Vert E_{\theta'}(s)\right\Vert _{2}\left\Vert E_{\theta}(s)-E_{\theta'}(s)\right\Vert _{2}ds\\
 & \leq & 2\sqrt{d}\overline{E_{\lambda}}\int_{0}^{t}\left\Vert E_{\theta}(s)-E_{\theta'}(s)\right\Vert _{2}ds
\end{array}
\]
by denoting $\overline{E_{\lambda}}=\sup_{t,\theta\in\left[0,\, T\right]\times\Theta}\left\Vert E_{\theta}(t)\right\Vert _{2}$
.

Now we bound the remaining term:
\[
\begin{array}{lll}
\int_{0}^{t}\left\Vert A_{\theta}(s)^{T}E_{\theta}(s)-A_{\theta'}(s)^{T}E_{\theta'}(s)\right\Vert _{2}ds & \leq & \sqrt{d}\int_{0}^{t}\left\Vert A_{\theta}(s)\right\Vert _{2}\left\Vert E_{\theta}(s)-E_{\theta'}(s)\right\Vert _{2}ds\\
 & + & \sqrt{d}\int_{0}^{t}\left\Vert E_{\theta'}(s)\right\Vert _{2}\left\Vert A_{\theta}(s)-A_{\theta'}(s)\right\Vert _{2}ds\\
 & \leq & \sqrt{d}\overline{A}\int_{0}^{t}\left\Vert E_{\theta}(s)-E_{\theta'}(s)\right\Vert _{2}\\
 & + & \sqrt{d}\overline{E_{\lambda}}\int_{0}^{t}\left\Vert A_{\theta}(s)-A_{\theta'}(s)\right\Vert _{2}d
\end{array}
\]

Using these bounds in the main inequality drive us to the following
inequality:

\[
\begin{array}{lll}
\left\Vert E_{\theta}(t)-E_{\theta'}(t)\right\Vert _{2} & \leq & 2\sqrt{d}(\frac{\overline{E_{\lambda}}}{\lambda}+\overline{A})\int_{0}^{t}\left\Vert E_{\theta}(s)-E_{\theta'}(s)\right\Vert _{2}ds\\
 & + & 2\sqrt{d}\overline{E_{\lambda}}\int_{0}^{t}\left\Vert A_{\theta}(s)-A_{\theta'}(s)\right\Vert _{2}ds
\end{array}
\]
then Gronwall's lemma gives us
\[
\left\Vert E_{\theta}(t)-E_{\theta'}(t)\right\Vert _{2}\leq2\sqrt{d}\overline{E_{\lambda}}\int_{0}^{T}\left\Vert A_{\theta}(s)-A_{\theta'}(s)\right\Vert _{2}dt.e^{\int_{0}^{t}2\sqrt{d}(\frac{\overline{E_{\lambda}}}{\lambda}+\overline{A})dt}
\]
and we obtain thanks to Cauchy-Schwarz inequality:
\[
\left\Vert E_{\theta}(t)-E_{\theta'}(t)\right\Vert _{2}\leq2\sqrt{dT}\overline{E_{\lambda}}\left\Vert A_{\theta}-A_{\theta'}\right\Vert _{L^{2}}e^{2\sqrt{d}(\frac{\overline{E_{\lambda}}}{\lambda}+\overline{A})T}
\]
which gives us the proper results using $\theta\longmapsto A_{\theta}$
continuity.

The bound for $\overline{E_{\lambda}}$ is obtained by a direct application
of Gronwall's lemma:

\[
\begin{array}{lll}
\left\Vert E_{\theta}(t)\right\Vert _{2} & \leq & \left|C\right|^{2}+\int_{0}^{t}\left\Vert A_{\theta'}(s)^{T}E_{\theta'}(s)+A_{\theta}(s)^{T}E_{\theta}(s)+\frac{1}{\lambda}E_{\theta}(s)^{T}E_{\theta}(s)\right\Vert _{2}ds\\
 & \leq & \left|C\right|^{2}+2\sqrt{d}\int_{0}^{t}\left(\left\Vert A_{\theta'}(s)\right\Vert _{2}+\frac{\overline{E_{\lambda}}}{\lambda}\right)\left\Vert E_{\theta'}(s)\right\Vert _{2}ds\\
 & \leq & \left|C\right|^{2}e^{2\sqrt{d}\left(\overline{A}+\frac{\overline{E_{\lambda}}}{\lambda}\right)T}
\end{array}
\]
hence 
\[
\overline{E_{\lambda}}\leq\left|C\right|^{2}e^{2\sqrt{d}\overline{A}T}e^{2\sqrt{d}\frac{\overline{E_{\lambda}}}{\lambda}}
\]
\end{proof}
\begin{lem}
\label{lem:h_discrepancy_lambda}Assuming condition C3 and C4 we know
it exists constants $K_{i}$ such that: \textup{ 
\[
\begin{array}{lll}
\left\Vert h_{\theta}(t,Y)-h_{\theta^{'}}(t,Y')\right\Vert _{2} & \leq & K_{6}e^{L_{1}\frac{\overline{E_{\lambda}}}{\lambda}}\left\Vert Y-Y'\right\Vert _{L^{2}}\\
 & + & \left(K_{7}+K_{8}\overline{E_{\lambda}}\right)\left(K_{4}+\frac{K_{5}}{\lambda}\overline{E_{\lambda}}e^{L_{1}\frac{\overline{E_{\lambda}}}{\lambda}}\right)e^{2L_{1}\frac{\overline{E_{\lambda}}}{\lambda}}\left\Vert \theta-\theta^{'}\right\Vert \\
 & + & \left(K_{9}e^{2L_{1}\frac{\overline{E_{\lambda}}}{\lambda}}+K_{10}e^{L_{1}\frac{\overline{E_{\lambda}}}{\lambda}}\right)\overline{E_{\lambda}}\left\Vert \theta-\theta^{'}\right\Vert 
\end{array}
\]
and 
\[
\left\Vert h_{\theta}(t,Y)\right\Vert _{2}\leq\sqrt{T}d^{2}\left\Vert C\right\Vert _{2}e^{\sqrt{d}\left(\overline{A}+\frac{\overline{E_{\lambda}}}{\lambda}\right)T}\left\Vert Y\right\Vert _{L^{2}}+\sqrt{d}\overline{E_{\lambda}}\overline{r_{\theta}}
\]
}\end{lem}
\begin{proof}
We know that $h_{\theta}(t,Y)=-\int_{0}^{t}R_{\theta}(t,s)C^{T}Y(s)ds-\int_{0}^{t}R_{\theta}(t,s)E_{\theta}(s)r_{\theta}(s)ds$
hence $\forall\left(Y,Y'\right)\in L^{\infty}(\left[0,\, T\right],\mathbb{R}^{d'})$
we have:

\[
\begin{array}{lll}
\left\Vert h_{\theta}(t,Y)-h_{\theta^{'}}(t,Y')\right\Vert _{2} & \leq & \sqrt{d}\left\Vert C\right\Vert _{2}\int_{0}^{t}\left\Vert R_{\theta}(t,s)\right\Vert _{2}\left\Vert Y(s)-Y'(s)\right\Vert _{2}ds\\
 & + & \sqrt{d}\left\Vert C\right\Vert _{2}\int_{0}^{t}\left\Vert R_{\theta}(t,s)-R_{\theta^{'}}(t,s)\right\Vert _{2}\left\Vert Y'(s)\right\Vert _{2}ds\\
 & + & d\overline{r}\int_{0}^{t}\left\Vert R_{\theta}(t,s)-R_{\theta'}(t,s)\right\Vert _{2}\left\Vert E_{\theta}(s)\right\Vert _{2}ds\\
 & + & d\overline{r}\int_{0}^{t}\left\Vert R_{\theta^{'}}(t,s)\right\Vert _{2}\left\Vert E_{\theta}(s)-E_{\theta^{'}}(s)\right\Vert _{2}ds\\
 & + & d\overline{r_{\theta}}\left\Vert \theta-\theta^{'}\right\Vert _{2}\int_{0}^{t}\left\Vert R_{\theta^{'}}(t,s)\right\Vert _{2}\left\Vert E_{\theta^{'}}(s)\right\Vert _{2}ds
\end{array}
\]
Cauchy Schwarz inequality gives us: 
\begin{equation}
\begin{array}{lll}
\left\Vert h_{\theta}(t,Y)-h_{\theta^{'}}(t,Y')\right\Vert _{2} & \leq & \sqrt{d}\left\Vert C\right\Vert _{2}\left\Vert R_{\theta}(.,s)\right\Vert _{L^{2}}\left\Vert Y-Y'\right\Vert _{L^{2}}\\
 & + & \left(\sqrt{d}\left\Vert C\right\Vert _{2}\left\Vert Y'\right\Vert _{L^{2}}+d\sqrt{T}\overline{r}\overline{E_{\lambda}}\right)\left\Vert R_{\theta}(.,s)-R_{\theta^{'}}(.,s)\right\Vert _{L^{2}}\\
 & + & d\overline{r}\left\Vert R_{\theta}(.,s)\right\Vert _{L^{2}}\left\Vert E_{\theta}-E_{\theta^{'}}\right\Vert _{L^{2}}\\
 & + & d\sqrt{T}\overline{r_{\theta}}\left\Vert \theta-\theta^{'}\right\Vert _{2}\left\Vert R_{\theta}(.,s)\right\Vert _{L^{2}}\overline{E_{\lambda}}
\end{array}\label{eq:h_inequality_interm}
\end{equation}
We straightforwardly bound $\left\Vert R_{\theta}(t,s)\right\Vert _{2}$
by application of Gronwall's lemma: 
\[
\begin{array}{lll}
\left\Vert R_{\theta}(t,s)\right\Vert _{2} & \leq & \sqrt{d}+\sqrt{d}\left(\overline{A}+\frac{\overline{E}}{\lambda}\right)\int_{s}^{t}\left\Vert R_{\theta}(u,s)\right\Vert _{2}du\\
 & \leq & \sqrt{d}e^{\sqrt{d}\left(\overline{A}+\frac{\overline{E}}{\lambda}\right)T}:=K_{3}e^{L_{1}\frac{\overline{E_{\lambda}}}{\lambda}}
\end{array}
\]
Using successively norm inequalities and Gronwall's lemma we obtain:
\[
\begin{array}{lll}
\left\Vert R_{\theta^{'}}(t,s)-R_{\theta}(t,s)\right\Vert _{2} & \leq & \sqrt{d}\left\Vert \alpha_{\theta}(t)-\alpha_{\theta^{'}}(t)\right\Vert _{2}\left\Vert R_{\theta}(t,s)\right\Vert _{2}\\
 & + & \sqrt{d}\left\Vert R_{\theta}(t,s)-R_{\theta^{'}}(t,s)\right\Vert _{2}\left\Vert \alpha_{\theta^{'}}(t)\right\Vert _{2}\\
 & \leq & \sqrt{d}\left(\overline{A_{\theta}}\left\Vert \theta-\theta^{'}\right\Vert _{2}+\frac{1}{\lambda}\left\Vert E_{\theta}(t)-E_{\theta^{'}}(t)\right\Vert _{2}\right)\left\Vert R_{\theta}(t,s)\right\Vert _{2}\\
 & + & \sqrt{d}\left\Vert R_{\theta}(t,s)-R_{\theta^{'}}(t,s)\right\Vert _{2}\left(\overline{A}+\frac{\overline{E_{\lambda}}}{\lambda}\right)\\
 & \leq & \sqrt{d}\left(\overline{A_{\theta}}+\frac{1}{\lambda}K_{1}\overline{E_{\lambda}}e^{L_{1}\frac{\overline{E_{\lambda}}}{\lambda}}\right)K_{3}e^{L_{1}\frac{\overline{E_{\lambda}}}{\lambda}}\left\Vert \theta-\theta^{'}\right\Vert \\
 & + & \sqrt{d}\left(\overline{A}+\frac{\overline{E_{\lambda}}}{\lambda}\right)\left\Vert R_{\theta}(t,s)-R_{\theta^{'}}(t,s)\right\Vert _{2}\\
 & \leq & \sqrt{d}e^{\sqrt{d}\overline{A}T}K_{3}\left(\overline{A_{\theta}}+\frac{1}{\lambda}K_{1}\overline{E_{\lambda}}e^{L_{1}\frac{\overline{E_{\lambda}}}{\lambda}}\right)e^{2L_{1}\frac{\overline{E_{\lambda}}}{\lambda}}\left\Vert \theta-\theta^{'}\right\Vert \\
 & := & \left(K_{4}+\frac{K_{5}}{\lambda}\overline{E_{\lambda}}e^{L_{1}\frac{\overline{E_{\lambda}}}{\lambda}}\right)e^{2L_{1}\frac{\overline{E_{\lambda}}}{\lambda}}\left\Vert \theta-\theta^{'}\right\Vert 
\end{array}
\]
and applying this bound in \ref{eq:h_inequality_interm} gives the
following inequality: 
\[
\begin{array}{lll}
\left\Vert h_{\theta}(t,Y)-h_{\theta^{'}}(t,Y')\right\Vert _{2} & \leq & \sqrt{d}\left\Vert C\right\Vert _{2}K_{3}e^{L_{1}\frac{\overline{E_{\lambda}}}{\lambda}}\left\Vert Y-Y'\right\Vert _{L^{2}}\\
 & + & \left(\sqrt{d}\left\Vert C\right\Vert _{2}\left\Vert Y'\right\Vert _{L^{2}}+d\sqrt{T}\overline{r}\overline{E_{\lambda}}\right)\left(K_{4}+\frac{K_{5}}{\lambda}\overline{E_{\lambda}}e^{L_{1}\frac{\overline{E_{\lambda}}}{\lambda}}\right)e^{2L_{1}\frac{\overline{E_{\lambda}}}{\lambda}}\left\Vert \theta-\theta^{'}\right\Vert \\
 & + & d\overline{r}K_{3}e^{L_{1}\frac{\overline{E_{\lambda}}}{\lambda}}K_{1}\overline{E_{\lambda}}e^{L_{1}\frac{\overline{E_{\lambda}}}{\lambda}}\left\Vert \theta-\theta^{'}\right\Vert \\
 & + & d\sqrt{T}\overline{r_{\theta}}K_{3}e^{L_{1}\frac{\overline{E_{\lambda}}}{\lambda}}\overline{E_{\lambda}}\left\Vert \theta-\theta^{'}\right\Vert _{2}\\
 & \leq & K_{6}e^{L_{1}\frac{\overline{E_{\lambda}}}{\lambda}}\left\Vert Y-Y'\right\Vert _{L^{2}}\\
 & + & \left(K_{7}+K_{8}\overline{E_{\lambda}}\right)\left(K_{4}+\frac{K_{5}}{\lambda}\overline{E_{\lambda}}e^{L_{1}\frac{\overline{E_{\lambda}}}{\lambda}}\right)e^{2L_{1}\frac{\overline{E_{\lambda}}}{\lambda}}\left\Vert \theta-\theta^{'}\right\Vert \\
 & + & \left(K_{9}e^{2L_{1}\frac{\overline{E_{\lambda}}}{\lambda}}+K_{10}e^{L_{1}\frac{\overline{E_{\lambda}}}{\lambda}}\right)\overline{E_{\lambda}}\left\Vert \theta-\theta^{'}\right\Vert 
\end{array}
\]
By a similar computation we obtain: 
\[
\left\Vert h_{\theta}(t,Y)\right\Vert _{2}\leq\sqrt{T}d^{2}\left\Vert C\right\Vert _{2}e^{\sqrt{d}\left(\overline{A}+\frac{\overline{E_{\lambda}}}{\lambda}\right)T}\left\Vert Y\right\Vert _{L^{2}}+\sqrt{d}\overline{E_{\lambda}}\overline{r_{\theta}}
\]
\end{proof}
\begin{lem}
\label{lem:E_inverse_discrepancy_wrt_lambda}Assuming condition C3
and C4 we know it exists constants $K_{i}$ such that: $\left\Vert E_{\theta}^{-1}(T)-E_{\theta'}^{-1}(T)\right\Vert _{2}\leq\left(\frac{K_{12}}{\lambda}+K_{11}\right)e^{K_{13}+\frac{K_{14}}{\lambda}}\left\Vert \theta-\theta'\right\Vert $
and $\left\Vert E_{\theta}^{-1}(T)\right\Vert _{2}\leq\frac{K_{15}}{\lambda}$\end{lem}
\begin{proof}
We have already shown$\forall\theta\in\Theta\: t\longmapsto E_{\theta}(t)$
is defined on $\left[0\,,\, T\right]$ and obviously continuous on
the same interval as an ODE solution.

When $E_{\theta}^{-1}(t)$ is defined we know it follows the ODE:

\[
\begin{array}{lll}
\frac{d}{dt}\left(E_{\theta}^{-1}(t)\right) & = & -E_{\theta}^{-1}(t)\dot{E_{\theta}}(t)E_{\theta}^{-1}(t)\\
 & = & -E_{\theta}^{-1}(t)\left(C^{T}C\text{\textminus}A_{\theta}(t)^{T}E_{\theta}(t)\text{\textminus}E_{\theta}(t)A_{\theta}(t)\text{\textminus}\frac{1}{\lambda}E_{\theta}(t)^{T}E_{\theta}(t)\right)E_{\theta}^{-1}(t)\\
 & = & \frac{1}{\lambda}+E_{\theta}^{-1}(t)A_{\theta}^{T}(t)+A_{\theta}(t)E_{\theta}^{-1}(t)-E_{\theta}^{-1}(t)C^{T}CE_{\theta}^{-1}(t)
\end{array}
\]
By hypothesis $\forall\theta\in\varTheta$ $E_{\theta}^{-1}(T)$ is
defined and by continuity of $\left(\theta,t\right)\longmapsto Det\left(E_{\theta}(t)\right)$
using chain rule we know for each $\theta$ it exists an interval
$\left[T-\varepsilon,\, T\right]$ and a open ball $B_{\varrho}\left(\theta\right)$
where $\left(t,\theta\right)\longmapsto E_{\theta}(t)$ is non-singular.
Because of $\theta\longmapsto E_{\theta}(T-\varepsilon)$ continuity
and differentiability it exists a constant $K_{11}$ such that: $\left\Vert E_{\theta}^{-1}(T-\epsilon)-E_{\theta'}^{-1}(T-\epsilon)\right\Vert _{2}\leq K_{11}\left\Vert \theta-\theta'\right\Vert $.

By defining $\overline{E_{\epsilon,\varrho,\theta,\lambda}^{-1}}=\sup_{\left[T-\varepsilon,\, T\right]\times B_{\varrho}\left(\theta\right)}\left\Vert E_{\theta}^{-1}(t)\right\Vert _{2}$
and using successively norm inequalities and Gronwall's lemma we obtain:
\[
\begin{array}{lll}
\left\Vert E_{\theta}^{-1}(t)-E_{\theta'}^{-1}(t)\right\Vert _{2} & \leq & \int_{T-\epsilon}^{t}\left\Vert E_{\theta}^{-1}(s)A_{\theta}^{T}(s)+A_{\theta}(s)E_{\theta}^{-1}(s)-E_{\theta^{'}}^{-1}(s)A_{\theta^{'}}^{T}(s)-A_{\theta^{'}}(s)E_{\theta^{'}}^{-1}(s)\right\Vert _{2}ds\\
 & + & \int_{T-\epsilon}^{t}\left\Vert E_{\theta^{'}}^{-1}(s)C^{T}CE_{\theta^{'}}^{-1}(s)-E_{\theta}^{-1}(s)C^{T}CE_{\theta}^{-1}(s)\right\Vert _{2}ds\\
 & + & K_{11}\left\Vert \theta-\theta'\right\Vert \\
 & \leq & 2\sqrt{d}\overline{A}\int_{T-\epsilon}^{t}\left\Vert E_{\theta}^{-1}(s)-E_{\theta^{'}}^{-1}(s)\right\Vert ds\\
 & + & 2\sqrt{d}\overline{E_{\theta,\lambda}^{-1}}\int_{T-\epsilon}^{t}\left\Vert A_{\theta}(s)-A_{\theta^{'}}^{T}(s)\right\Vert _{2}ds\\
 & + & 2d^{\frac{3}{2}}\overline{E_{\theta,\lambda}^{-1}}\left\Vert C\right\Vert _{2}^{2}\int_{T-\epsilon}^{t}\left\Vert E_{\theta^{'}}^{-1}(s)-E_{\theta}^{-1}(s)\right\Vert _{2}ds\\
 & + & K_{11}\left\Vert \theta-\theta'\right\Vert \\
 & \leq & \left(2\sqrt{d}\overline{A}+2d^{\frac{3}{2}}\overline{E_{\epsilon,\varrho,\theta,\lambda}^{-1}}\left\Vert C\right\Vert _{2}^{2}\right)\int_{T-\epsilon}^{t}\left\Vert E_{\theta}^{-1}(s)-E_{\theta^{'}}^{-1}(s)\right\Vert ds\\
 & + & \left(2\sqrt{dT}\overline{E_{\epsilon,\varrho,\theta,\lambda}^{-1}}\overline{A_{\theta}}+K_{11}\right)\left\Vert \theta-\theta'\right\Vert \\
 & \leq & \left(2\sqrt{dT}\overline{E_{\epsilon,\varrho,\theta,\lambda}^{-1}}\overline{A_{\theta}}+K_{11}\right)e^{\left(2\sqrt{d}\overline{A}+2d^{\frac{3}{2}}\overline{E_{\epsilon,\varrho,\theta,\lambda}^{-1}}\left\Vert C\right\Vert _{2}^{2}\right)T}\left\Vert \theta-\theta'\right\Vert 
\end{array}
\]
For $\left\Vert E_{\theta}^{-1}(T)\right\Vert _{2}$ we can obtain
a uniform bound w.r.t $\theta$ using Gronwall's lemma: 
\[
\begin{array}{lll}
\left\Vert E_{\theta}^{-1}(t)\right\Vert _{2} & \leq & \int_{T-\varepsilon}^{t}\left\Vert \frac{1}{\lambda}+E_{\theta}^{-1}(s)A_{\theta}^{T}(s)+A_{\theta}(s)E_{\theta}^{-1}(s)-E_{\theta}^{-1}(t)C^{T}CE_{\theta}^{-1}(t)\right\Vert _{2}ds\\
 & \leq & \frac{1}{\lambda}+\int_{T-\varepsilon}^{t}\sqrt{d}\left(2\overline{A}+d\overline{E_{\epsilon,\varrho,\theta,\lambda}^{-1}}\left\Vert C\right\Vert _{2}^{2}\right)\left\Vert E_{\theta}^{-1}(t)\right\Vert _{2}^{2}\\
 & \leq & \frac{1}{\lambda}e^{\sqrt{d}\left(2\overline{A}+d\overline{E_{\epsilon,\varrho,\theta,\lambda}^{-1}}\left\Vert C\right\Vert _{2}^{2}\right)T}\\
 & := & \frac{K_{15}}{\lambda}
\end{array}
\]
by using this upper bound in the previous inequality we obtain the
desired result. 
\end{proof}

\section{\label{sec:Gradient-Computation}Gradient Computation}

For optimization purpose we need to compute the gradient of $S(\widehat{Y};\theta,\lambda)$
.

\subsection{Notation in row vector for the adjoint ODE vector field}

We will define the solution of the adjoint ODE in row formulation,
we introduce 
\[
Q_{\theta}(t)=\left(\widehat{h_{\theta}}^{T},\left(E_{\theta}^{r}\right)^{T}\right)^{T}(t)
\]
with $E_{\theta}^{r}:=\left(E_{\theta,1}^{T},\cdots,E_{\theta,d}^{T}\right)^{T}$the
row formulation of $E_{\theta}$, $E_{\theta,i}$ beeing the $i-th$
column of $E_{\theta}$. It is a $D:=d^{2}+d$ sized function respecting
the ODE : 
\[
\begin{array}{l}
\dot{Q_{\theta}}=F(Q_{\theta},\theta,t)\\
Q_{\theta}(0)=\left(\begin{array}{cc}
0_{1,d} & Q^{r}\end{array}\right)^{T}
\end{array}
\]
by introducing $Q^{r}:=\left(Q_{1}^{T},\cdots,Q_{d}^{T}\right)^{T}$
the row formulation of $Q$ and the general vector field $F$:

\[
F(Q_{\theta},\theta,t)=\left(\begin{array}{c}
G(Q_{\theta},\theta,t)\\
H(Q_{\theta},\theta)
\end{array}\right)
\]
with $G$ and $H$ defined by: 
\[
\begin{array}{lll}
G(Q_{\theta},\theta,t) & := & -\left(A_{\theta}{}^{T}+\frac{E_{\theta}}{\lambda}\right)\widehat{h_{\theta}}-\left(C^{T}Y+E_{\theta}r_{\theta}\right)\\
H_{(j-1)d+i}(Q_{\theta},\theta) & := & \delta_{i,j}-(A_{\theta,i}^{T}E_{j}+A_{\theta,j}^{T}E_{\theta,i}+\frac{1}{\lambda}E_{\theta,i}^{T}E_{\theta,j})
\end{array}
\]
and $A_{\theta,i}$ beeing the $i-th$ column of $A_{\theta}$.

For the next subsections we will drop dependence in $\theta$ for
$A_{\theta},\, r_{\theta},\, E_{\theta},\,\widehat{h_{\theta}}$

\subsection{Gradient computation by sensitivity equation }

Straightforward computation gives us : 
\[
\begin{array}{lll}
\nabla_{\theta}S(\widehat{Y};\theta,\lambda) & = & -2\int_{0}^{T}\left(\frac{\partial r(t)}{\partial\theta}^{T}\widehat{h}(t)+\frac{\partial\widehat{h}(t)}{\partial\theta}^{T}r(t)+\frac{1}{\lambda}\frac{\partial\widehat{h}(t)}{\partial\theta}^{T}\widehat{h}(t)\right)dt\\
 & - & \frac{\partial\widehat{h}(T)}{\partial\theta}^{T}E(T)^{-1}\widehat{h}(T)\\
 & - & \left(\widehat{h}(T)^{T}\frac{\partial\left(E(T)_{i}^{-1}\right)}{\partial\theta}+\left(E(T)_{i}^{-1}\right)^{T}\frac{\partial\widehat{h}(T)}{\partial\theta}\right)_{1\leq i\leq d}^{T}\widehat{h}(T)
\end{array}
\]
with: 
\[
\frac{\partial\left(E(T)^{-1}\right)}{\partial\theta_{j}}=-E(T)^{-1}\left(\frac{\partial E(T)}{\partial\theta_{j}}\right)E(T)^{-1}
\]
thus we need to compute $\frac{\partial Q_{\theta}(t),}{\partial\theta}$
solution of the sensitivity equation: 
\[
\frac{d}{dt}(\frac{\partial Q_{\theta}(t)}{\partial\theta})=\frac{\partial F}{\partial Q}(Q_{\theta}(t),\theta,t)\frac{\partial Q_{\theta}(t)}{\partial\theta}+\frac{\partial F}{\partial\theta}(Q_{\theta}(t),\theta,t)
\]
and we know that $R_{\theta}(0)=(0,Q^{r})$ so $\frac{\partial Q_{\theta}(0)}{\partial\theta}=0$,
hence we can obtain $\frac{\partial Q_{\theta}(t)}{\partial\theta}$
by solving the Cauchy problem:

\[
\begin{array}{l}
\frac{d}{dt}(\frac{\partial Q_{\theta}(t)}{\partial\theta})=\frac{\partial F}{\partial Q}(Q_{\theta}(t),\theta,t)\frac{\partial Q_{\theta}(t)}{\partial\theta}+\frac{\partial F}{\partial\theta}(Q_{\theta}(t),\theta,t)\\
\frac{\partial Q_{\theta}(0)}{\partial\theta}=0
\end{array}
\]

In order to compute sensitivity equation we need to compute $\frac{\partial F}{\partial Q}(Q_{\theta},\theta,t)$
and $\frac{\partial F}{\partial\theta}(Q_{\theta},\theta,t)$, for
$\frac{\partial F}{\partial Q}(Q_{\theta},\theta,t)$ and $\frac{\partial F}{\partial\theta}(Q_{\theta},\theta,t)$
we obtain:

\[
\begin{array}{l}
\frac{\partial F}{\partial Q}(Q_{\theta},\theta,t)=\left(\begin{array}{cc}
-\left(A(t)^{T}+\frac{E}{\lambda}\right) & \frac{\partial G_{i}}{\partial E_{j}^{r}}(Q_{\theta},\theta,t)\\
0_{d^{2},d} & \frac{\partial H(Q_{\theta},\theta)}{\partial E^{r}}
\end{array}\right)\\
\frac{\partial F}{\partial\theta}(Q_{\theta},\theta,t)=\left(\begin{array}{c}
\frac{\partial G}{\partial\theta}(Q_{\theta},\theta,t)\\
\frac{\partial H}{\partial\theta}(E^{r},\theta)
\end{array}\right)
\end{array}
\]
with: 
\[
\begin{array}{l}
\frac{\partial G_{i}}{\partial E_{(k-1)d+h}^{r}}(Q_{\theta},\theta,t)=-\delta_{i,h}\left(r(t)+\frac{h}{\lambda}\right)_{k}\\
\frac{\partial G}{\partial\theta}(Q_{\theta},\theta,t)=-\left(h^{T}\frac{\partial A_{i}(t)}{\partial\theta}\right)_{1\leq i\leq d}-E\frac{\partial r(t)}{\partial\theta}
\end{array}
\]

We also need to compute $H(Q_{\theta},\theta)$ partial derivative
w.r.t $E^{r}$ and $\theta$, we have: 
\[
\left(\frac{\partial H(E^{r},\theta)}{\partial E^{r}}\right)_{(j-1)d+i}=-\left(\begin{array}{ccccc}
0 & A_{j}^{t} & 0 & A_{i}^{t} & 0\end{array}\right)-\frac{1}{\lambda}\left(\begin{array}{ccccc}
0 & E_{j}^{t} & 0 & E_{i}^{t} & 0\end{array}\right)
\]
because: 
\begin{itemize}
\item $\frac{\partial}{\partial E^{r}}\left(A_{j}^{t}E_{i}+A_{i}^{t}E_{j}\right)=\left(\begin{array}{ccccc}
0 & A_{j}^{t} & 0 & A_{i}^{t} & 0\end{array}\right)$ where $A_{j}^{t}$ is in $i-th$ position and $A_{i}^{t}$ is in
$j-th$ position. 
\item $\frac{1}{\lambda}\frac{\partial}{\partial E}\left(E_{j}^{t}E_{i}\right)=\left(\begin{array}{ccccc}
0 & \frac{1}{\lambda}E_{j}^{t} & 0 & 0 & 0\end{array}\right)+\left(\begin{array}{ccccc}
0 & 0 & \frac{1}{\lambda}E_{i}^{t} & 0 & 0\end{array}\right)$ where $E_{j}^{t}$ is in $i-th$ position and $E_{i}^{t}$ is in
$j-th$ position. 
\end{itemize}
and:

\[
\left(\frac{\partial H(Q_{\theta},\theta)}{\partial\theta}\right)_{(j-1)d+i}=-E_{i}^{t}\frac{\partial A_{j}}{\partial\theta}-E_{j}^{t}\frac{\partial A_{i}}{\partial\theta}
\]
because:
\begin{itemize}
\item $\frac{\partial}{\partial\theta}\left(A_{j}^{t}E_{i}+A_{i}^{t}E_{j}\right)=E_{i}^{t}\frac{\partial A_{j}}{\partial\theta}+E_{j}^{t}\frac{\partial A_{i}}{\partial\theta}$
where $\frac{\partial A_{i}}{\partial\theta}=\left(\frac{\partial A_{i}}{\partial\theta_{1}}\cdots\frac{\partial A_{i}}{\partial\theta_{p}}\right)$
a $d\times p$ matrix 
\end{itemize}

{\footnotesize{}\bibliographystyle{plain}
\bibliography{biblio_kalman_ode}

\begin{thebibliography}{10}

\bibitem{Bardi2008}
M.~Bardi and I.~Capuzzo-Dolcetta.
\newblock {\em Optimal Control and Viscosity Solutions of
  Hamilton-Jacobi-Bellman Equations}.
\newblock Modern Birkhauser Classic. Birkhauser, 2008.

\bibitem{Bellman1970}
R.~Bellman and K.J Astrom.
\newblock On structural identifiability.
\newblock {\em Mathematical Biosciences}, 7:329--339, 1970.

\bibitem{Bensoussan2004}
A.~Bensoussan.
\newblock {\em Stochastic Control of Partially Observable Systems}.
\newblock Cambridge University Press, 2004.

\bibitem{BickelRitov2003}
P.J. Bickel and Y.~Ritov.
\newblock Nonparametric estimators which can be plugged-in.
\newblock {\em Annals of Statistics}, 31(4):4, 2003.

\bibitem{Brunel2008}
N.~J-B. Brunel.
\newblock Parameter estimation of ode's via nonparametric estimators.
\newblock {\em Electronic Journal of Statistics}, 2:1242--1267, 2008.

\bibitem{Brunel2014}
N.~J-B. Brunel and Q.~Clairon~F. D'Alche-Buc.
\newblock Parameter estimation of ordinary differential equations with
  orthogonality conditions.
\newblock {\em JASA}, 109:173--185, 2014.

\bibitem{BrunelClairon_Riccati2014}
Nicolas J-B. Brunel and Quentin Clairon.
\newblock A tracking approach to parameter estimation in linear ordinary
  differential equations.
\newblock Technical report, 2014.
\newblock submitted.

\bibitem{Lyzell2011}
M.~Enqvist C.~Lyzell, T.~Glad and L.~Ljung.
\newblock Difference algebra and system identification.
\newblock {\em Automatica}, 47:1896--1904, 2011.

\bibitem{calderhead2009}
B.~Calderhead, M~Girolami, and N.D Lawrence.
\newblock Accelerating bayesian inference over nonlinear differential equations
  with gaussian processes.
\newblock In {\em Advances in Neural Information Processing Systems 21 -
  Proceedings of the 2008 Conference}, 2009.

\bibitem{Campbell2013}
D.A. Campbell and O.~Chkrebtii.
\newblock Maximum profile likelihood estimation of differential equation
  parameters through model based smoothing state estimates.
\newblock {\em Mathematical Biosciences}, 2013.

\bibitem{Chernovena2014}
I.~Chernovena, B.~Freydin, B.~Hipszer, and Apanasovich~T. V.
\newblock Estimation of nonlinear differential equation for glucose-insulin
  dynamics in type i diabetic patients using generalized smoothing.
\newblock {\em Annals of Applied Statistics}, 8(2):886--904, 2014.

\bibitem{Campbell2011}
G.~Hooker D.A.~Campbell and K.~B. McAuley.
\newblock Parameter estimation in differential equation models with constrained
  states.
\newblock {\em Journal of Chemometrics}, 26:322--332, 2011.

\bibitem{Ellner2006}
S.~P. Ellner and J.~Guckenheimer.
\newblock {\em Dynamic Models in Biology, Number vol. 13 in Princeton
  Paperbacks}.
\newblock Princeton, NJ: Princeton University Press, 2006.

\bibitem{Engl2009}
Hein~W Engl, Christoph Flamm, Philipp K{\"u}gler, James Lu, Stefan M{\"u}ller,
  and Peter Schuster.
\newblock Inverse problems in systems biology.
\newblock {\em Inverse Problems}, 25(12), 2009.

\bibitem{gelman1996}
A.~Gelman, F.~Bois, and J.~Jiang.
\newblock Physiological pharmacokinetic analysis using population modeling and
  informative prior distributions.
\newblock {\em Journal of the American Statistical Association}, 91, 1996.

\bibitem{ghasemi2011}
O.~Ghasemi, M.~Lindsey, T.~Yang, N.~Nguyen, Y.~Huang, and Y.~Jin.
\newblock Bayesian parameter estimation for nonlinear modelling of biological
  pathways.
\newblock {\em BMC Systems Biology}, 5, 2011.

\bibitem{GugushviliKlaassen2010}
S.~Gugushvili and C.A.J. Klaassen.
\newblock {Root-n-consistent parameter estimation for systems of ordinary
  differential equations: bypassing numerical integration via smoothing}.
\newblock {\em Bernoulli}, to appear, 2011.

\bibitem{Happel1980}
J.~Happel, I.~Suzuki, P.~Kokayeff, and V.~Fthenakis.
\newblock Multiple isotope tracing of methanation over nickel catalyst.
\newblock {\em Journal of Catalysis}, 65:59--77, 1980.

\bibitem{Hooker2009}
G.~Hooker.
\newblock Forcing function diagnostics for nonlinear dynamics.
\newblock {\em Biometrics}, 65:928--936, 2009.

\bibitem{HookerEllner2013}
G.~Hooker and S.~P. Ellner.
\newblock Goodness of fit in nonlinear dynamics: Mis-specified rates or
  mis-specified states?
\newblock arxiv:1312.0294., arXiv preprint, 2013.

\bibitem{Huang2006}
Y.~Huang and H.~Wu.
\newblock A bayesian approach for estimating antiviral efficacy in hiv dynamic
  models.
\newblock {\em Journal of Applied Statistics}, 33:155--174, 2006.

\bibitem{Hubert1999}
E.~Hubert.
\newblock Essential components of an algebraic differential equation.
\newblock {\em J. Symbolic Computation}, 28:657--680, 1999.

\bibitem{Vidal2000}
C.~Noiret L.~Denis-Vidal, G. Joly-Blanchart.
\newblock Some effective approaches to check the identifiability of
  uncontrolled nonlinear systems.
\newblock {\em Mathematics and Computers in Simulation}, 57:35--44, 2000.

\bibitem{LiOsborne2005}
Z.~Li, M.R. Osborne, and T.~Prvan.
\newblock Parameter estimation of ordinary differential equations.
\newblock {\em IMA Journal of Numerical Analysis}, 25:264--285, 2005.

\bibitem{LiangWujasa2008}
H~Liang and H.~Wu.
\newblock Parameter estimation for differential equation models using a
  framework of measurement error in regression models.
\newblock {\em Journal of the American Statistical Association},
  103(484):1570--1583, December 2008.

\bibitem{Marbach2012}
Costello J. C.-Küffner R. Vega N. M. Prill R. J. Camacho D. M. ... \&
  DREAM5~Consortium. Marbach, D.
\newblock Wisdom of crowds for robust gene network inference.
\newblock {\em Nature methods}, 9(8):796--804., 2012.

\bibitem{MiaoXiaPerelsonWu-siam2011}
H.~Miao, X.~Xia, A.~S. Perelson, and H.~Wu.
\newblock On identifiability of nonlinear ode models and applications in viral
  dynamics.
\newblock {\em SIAM Review}, 53:3--39, 2011.

\bibitem{Newey1997}
W.~K. Newey.
\newblock Convergence rates and asymptotic normality for series estimators.
\newblock {\em Journal of Econometrics}, 79:147--168, 1997.

\bibitem{Nowak2000}
M.A. Nowak and R.M. May.
\newblock {\em Virus Dynamics: Mathematical Principles of Immunology and
  Virology}.
\newblock Oxford University Press, 2000.

\bibitem{Pohjanpalo1978}
H.~Pohjanpalo.
\newblock System identifiability based on the power series expansion of the
  solution.
\newblock {\em Mathematical Biosciences}, 41:21--33, 1978.

\bibitem{Pronzato2008}
L.~Pronzato.
\newblock Optimal experimental design and some related control problems.
\newblock {\em Automatica}, 44:303--325, 2008.

\bibitem{QiZhao2010}
Xin Qi and Hongyu Zhao.
\newblock Asymptotic efficiency and finite-sample properties of the generalized
  profiling estimation of parameters in ordinary differential equations.
\newblock {\em The Annals of Statistics}, 1:435--481, 2010.

\bibitem{Ramsay2007}
J.O. Ramsay, G.~Hooker, J.~Cao, and D.~Campbell.
\newblock Parameter estimation for differential equations: A generalized
  smoothing approach.
\newblock {\em Journal of the Royal Statistical Society (B)}, 69:741--796,
  2007.

\bibitem{ruppert2003semiparametric}
D.~Ruppert, M.P. Wand, and R.J. Carroll.
\newblock {\em {Semiparametric regression}}.
\newblock Cambridge series on statistical and probabilistic mathematics.
  Cambridge University Press, 2003.

\bibitem{Sontag1998}
E.~Sontag.
\newblock {\em Mathematical Control Theory: Deterministic finite-dimensional
  systems}.
\newblock Springer-Verlag (New-York), 1998.

\bibitem{Stuart2010}
A.M. Stuart.
\newblock Inverse problems: A bayesian perspective.
\newblock {\em Acta Numerica}, pages 451--559, 2010.

\bibitem{Vaart1998}
A.W. van~der Vaart.
\newblock {\em Asymptotic Statistics}.
\newblock Cambridge Series in Statistical and Probabilities Mathematics.
  Cambridge University Press, 1998.

\bibitem{Varah1982}
J.~M. Varah.
\newblock A spline least squares method for numerical parameter estimation in
  differential equations.
\newblock {\em SIAM J.sci. Stat. Comput.}, 3(1):28--46, 1982.

\bibitem{walter1997identification}
Eric Walter and Luc Pronzato.
\newblock {\em Identification of parametric models}.
\newblock Communications and Control Engineering. Springer Verlag New-York,
  1997.

\bibitem{Wu2014}
H.~Wu, T.~Lu, H.~Xue, and H.~Liang.
\newblock Sparse additive odes for dynamic gene regulatory network modeling.
\newblock {\em Journal of the American Statistical Association},
  109(506):700--716, 2014.

\end{thebibliography}
}
\end{document}